\documentclass[aip,graphicx]{revtex4-1}

\draft

\usepackage[T1]{fontenc}
\usepackage[utf8]{inputenc}
\usepackage{graphicx}
\usepackage[dvipsnames]{xcolor}
\usepackage{tgtermes}
\usepackage{centernot}
\usepackage{tikz}
\usetikzlibrary{positioning}
\usepackage[
pdftitle={Moments of logarithmic derivative SO, USp}, 
pdfauthor={E. Alvarez},
colorlinks=true,linkcolor=blue,urlcolor=blue,citecolor=blue,bookmarks=true,
bookmarksopenlevel=2]{hyperref}
\usepackage{amsmath,amssymb,amsthm,textcomp, tabu}
\usepackage{enumerate}

\newlength{\bibitemsep}
\newlength{\bibparskip}
\let\oldthebibliography\thebibliography
\renewcommand\thebibliography[1]{%
  \oldthebibliography{#1}%
  \setlength{\parskip}{\bibitemsep}%
  \setlength{\itemsep}{\bibparskip}%
}

\usepackage{geometry}
\newcommand{\abs}[1]{\left\lvert #1 \right\rvert}

\newtheorem{theorem}{Theorem}
\newtheorem{prop}{Proposition}[section]
\newtheorem{lemma}[prop]{Lemma}

\theoremstyle{definition}
\newtheorem{defi}[prop]{Definition}

\theoremstyle{remark}
\newtheorem{rem}[prop]{Remark}

\begin{document}
\title{Moments of the logarithmic derivative of characteristic polynomials from $SO(N)$ and $USp(2N)$}

\author{E. Alvarez}  
\email{emilia.alvarez@bristol.ac.uk}
\author{N.C. Snaith }
\email[Electronic mail (Corresponding author): ]{n.c.snaith@bris.ac.uk}
\affiliation{School of Mathematics, University of Bristol, UK}


\begin{abstract}
    We study moments of the logarithmic derivative of characteristic polynomials of orthogonal and symplectic random matrices. In particular, we compute the asymptotics for large matrix size, $N$, of these moments evaluated at points which are approaching 1. This follows work of Bailey, Bettin, Blower, Conrey, Prokhorov, Rubinstein and Snaith where they compute these asymptotics in the case of unitary random matrices. 
\end{abstract}

\maketitle

\section{Introduction}
This work is based on recent results by Bailey, Bettin, Blower, Conrey, Prokhorov, Rubinstein and Snaith \cite{MM}, which considers mixed moments of the characteristic polynomial and its derivative, as well as moments of the logarithmic derivative of characteristic polynomials, averaged over $U(N)$ with Haar measure (also called the CUE). Computations of this nature are motivated by the striking resemblance between random matrix models and their number theoretic analogues, see for example \cite{KS,kn:hughes03,CFKRS} or the review papers \cite{Conrey,kn:snaith10,kn:keasna03}. The relationship between  characteristic polynomials and the Riemann zeta-function via their statistical properties has been extended to the study of families of $L$-functions on the number theory side modeled by unitary, orthogonal and symplectic matrix ensembles based on symmetry \cite{KatzS,kn:katzsarnak99b,kn:confar00, kn:keasna00b, kn:cfz2}. A variety of moments on the random matrix side have been studied in how they relate to moments of the Riemann zeta-function, $\zeta(s)$.  Moments of the logarithmic derivative of the Riemann zeta-function have been studied since Selberg, who (assuming the Riemann Hypothesis, RH) used bounds on the second moment of $\frac{\zeta^{'}}{\zeta}(s)$, averaging on a vertical line in the complex plane up to height $T$ and just off the critical line (but approaching it), to study primes in short intervals \cite{Selberg}. Goldston, Gonek and Montgomery, also on RH, further showed that an asymptotic form of this moment is equivalent to the pair correlation conjecture for the zeroes of the Riemann zeta-function, and also linked this to counting prime powers \cite{GGM}. Farmer et. al extended this work to give equivalences between mean values of products of logarithmic derivatives of zeta, higher correlation functions for the zeroes of zeta and integers that are products of a fixed number of prime powers, which they call almost primes \cite{FGLL}. In their paper, by assuming the random matrix conjectures (that the correlation functions of zeroes of $\zeta(s)$ agree, in appropriate scaling limits,  with the correlation functions of eigenvalues of Gaussian Unitary matrices), they give explicit computations based on the work of Conrey and Snaith \cite{ConSna}. Farmer also studied a mixed moment of the logarithmic derivative of the Riemann zeta-function assuming RH and random matrix conjectures \cite{Farmer}. Finally, the distribution of the logarithmic derivative of $\zeta(s)$ has also been studied just off the critical line by Guo \cite{Guo} and Lester \cite{Lester}. The latter showed that this distribution converges to a two-dimensional Gaussian distribution in the complex plane, which agrees with Guo's earlier work. 
 
Of particular interest for our current work on the random matrix side, Hughes, Keating and O'Connell  \cite{HKO}, Mezzadri  \cite{Mez} and Conrey, Rubinstein and Snaith  \cite{CRS} study moments of derivatives of characteristic polynomials, and Hughes \cite{kn:hug01}, Dehaye \cite{Dehayederiv}, Riedtmann  \cite{HR} and Winn  \cite{Winn} discuss joint moments of characteristic polynomials and their derivatives. 

In this work we derive in the orthogonal and symplectic cases the analogue of Theorem 1.2 from the work of Bailey et al. \cite{MM}, which computes the moments of the logarithmic derivative of characteristic polynomials averaged over the unitary group. We see some curious differences in the leading order term in the limit as matrix size, $N$, grows to infinity and the point at which the characteristic polynomial is evaluated, $e^{-\alpha}$ tends towards 1 on the unit circle.  These differences are evident in the dependence on $N$ and $a=N\alpha$ in the theorems set out below. 

We define the characteristic polynomial of a matrix $X$ to be 
\begin{equation}
    \Lambda_X(s) = \det(I - sX^{*}) ,
\end{equation}
where $X^*$ is the conjugate transpose of $X$. 
The eigenvalues of an even orthogonal $SO(2N)$ or unitary symplectic $USp(2N)$ matrix $X$ come in conjugate pairs $e^{i\theta_1},e^{-i\theta_1} \dots, e^{i\theta_N}, e^{-i\theta_N}$ and so its characteristic polynomial can also be expressed as:
\begin{equation}
   \Lambda_X(s) = \prod_{j=1}^N (1 - se^{-i\theta_j})(1 - se^{i\theta_j}).
\end{equation}
We evaluate moments using the Haar measure, denoted $\mathrm{dX}$ and integrating over the whole matrix ensemble, $SO(2N)$, $USp(2N)$ or $SO(2N+1)$.  The asymptotic computation of Bailey et al. \cite{MM} considers even integer moments of the logarithmic derivative of characteristic polynomials of random unitary matrices, evaluated with respect to the Haar measure. We extend this theorem to the orthogonal and symplectic ensembles.

\begin{theorem}{(Theorem 1.2 from Bailey et al., \cite{MM})}\label{thm1.2}
Let $\Lambda_X$ denote the characteristic polynomial of a matrix $X \in U(N)$, $\mathfrak{Re}(a)>0$ and $K \in \mathbb{N}$. Then,
\begin{equation}
    \int_{U(N)} \abs{\frac{\Lambda_X^{'}}{\Lambda_X}(e^{-\alpha})}^{2K} \mathrm{dX} = \binom{2K-2}{K-1}\frac{N^{2K}}{(2a)^{2K-1}}\left(1+\mathcal{O}(a)\right),
\end{equation} where $\alpha = a/N$ and $a = o(1)$ as $N \to \infty$. 
\end{theorem}

Our main results are below. Theorem \ref{thm2} extends Theorem \ref{thm1.2} to the even orthogonal ensemble and the proof is given in detail in Section \ref{s2}. Theorem \ref{thm3} applies to the symplectic case which is proved in Section \ref{Symplectic} and Theorem \ref{thm4} applies to the odd orthogonal ensemble, with the proof outlined briefly in Section \ref{odd}. Section \ref{prems} contains preliminary results that we will call upon, and Section \ref{theset} describes the properties of a set of matrices that arises in our calculations.  We derive several general results in Section \ref{theset}, which we then use in the proofs of Theorems \ref{thm2}, \ref{thm3} and \ref{thm4}.

\begin{theorem}\label{thm2}
  Let $\Lambda_X(s)$ denote the characteristic polynomial of a matrix $X \in SO(2N)$, the group of even dimensional random orthogonal matrices with determinant 1 equipped with the Haar measure $\mathrm{dX}$. Let $K \in \mathbb{N}$, $\alpha = a/N$ where $a = o(1)$ as $N \to \infty$ and $\mathfrak{Re}(a)>0$. Then, as $N$ tends to $\infty$ and for $K \geq 2$, the moments of the logarithmic derivative of $\Lambda_X(s)$ evaluated at $e^{-\alpha}$ are given by:
  \begin{equation}
      \int_{SO(2N)} \left(\frac{\Lambda_X^{'}}{\Lambda_X}(e^{-\alpha})\right)^{K} \mathrm{dX} = (-1)^K\frac{2N^K}{a^{K-1}}  \frac{(2K-3)!!}{(K-1)!}  \left( 1 + \mathcal{O}(a) \right).     
  \end{equation}
  The first moment is given by 
\begin{equation}
    \int_{SO(2N)} \left(\frac{\Lambda_X^{'}}{\Lambda_X}(e^{-\alpha})\right)^{1} \mathrm{dX} = -N  \left( 1 + \mathcal{O}(a) \right).
\end{equation}
\end{theorem}
\begin{theorem}\label{thm3}
  Let $\Lambda_X(s)$ denote the characteristic polynomial of a matrix $X \in USp(2N)$, the group of even dimensional random unitary symplectic matrices equipped with the Haar measure $\mathrm{dX}$. Let $K \in \mathbb{N}$, $\alpha = a/N$ where $a = o(1)$ as $N \to \infty$ and $\mathfrak{Re}(a)>0$. Then, as $N$ tends to $\infty$ and for $K \geq 4$, the moments of the logarithmic derivative of $\Lambda_X(s)$ evaluated at $e^{-\alpha}$ are given by: 
\begin{align*}
    \int_{USp(2N)} \left( \frac{ \Lambda^{'}_X}{\Lambda_X}(e^{-\alpha}) \right)^K  \, \mathrm{dX}  
   =   (-1)^K \frac{2}{3}\frac{N^K}{a^{K-3}} \frac{(2K-5)!!}{(K-1)!}  \left(1 + \mathcal{O}(a) \right).
\end{align*}
The first three moments are given by:
\begin{align*}
\int_{USp(2N)} \left( \frac{ \Lambda^{'}_X}{\Lambda_X}(e^{-\alpha}) \right)^1  \, \mathrm{dX}  = N\left(1 + \mathcal{O}(a)\right), \\
\int_{USp(2N)} \left( \frac{ \Lambda^{'}_X}{\Lambda_X}(e^{-\alpha}) \right)^2  \, \mathrm{dX} = N^2\left(1 + \mathcal{O}(a)\right),\\
\int_{USp(2N)} \left( \frac{ \Lambda^{'}_X}{\Lambda_X}(e^{-\alpha}) \right)^3  \, \mathrm{dX}   = \frac{2}{3} N^3 \left(1 + \mathcal{O}(a) \right).
\end{align*}
\end{theorem}

\begin{theorem}\label{thm4}
  Let $\Lambda_X(s)$ denote the characteristic polynomial of a matrix $X \in SO(2N+1)$, the group of odd dimensional random orthogonal matrices with determinant 1 equipped with the Haar measure $\mathrm{dX}$. Let $K \in \mathbb{N}$, $\alpha = a/N$ where $a = o(1)$ as $N \to \infty$ and $\mathfrak{Re}(a)>0$. Then, as $N$ tends to $\infty$, the moments of the logarithmic derivative of $\Lambda_X(s)$ evaluated at $e^{-\alpha}$ are given by:  
  \begin{align}
    \int_{SO(2N+1)} \left(\frac{ \Lambda_X^{'}}{\Lambda_X }(e^{-\alpha})\right)^K \mathrm{dX} \nonumber \\
     = (-1)^K \left[ \left( \frac{N}{a} \right)^K - \frac{N^K}{a^{K-1}}K \right] + \mathcal{O}\left(\frac{N^{K-1}}{a^{K-1}}\right)  + \mathcal{O}\left( \frac{N^K}{a^{K-2}}  \right). 
\end{align}
\end{theorem}

\section{Interpretation of the results}\label{interpretation}
The leading order behaviour of moments of the logarithmic derivative of characteristic polynomials taken from $SO(2N+1), SO(2N)$ and $USp(2N)$ is largely governed by the likelihood of a matrix in each respective ensemble having an eigenvalue at or near $1$.  Since the logarithmic derivative has the characteristic polynomial in the denominator, an eigenvalue at 1 causes a singularity in the logarithmic derivative when evaluated at $s=1$. Below we've included the one-level density in the large $N$ limit for the four ensembles considered, for reference. 

\begin{figure}[htbp]
\begin{center}
\includegraphics[scale=.25,angle=-90]{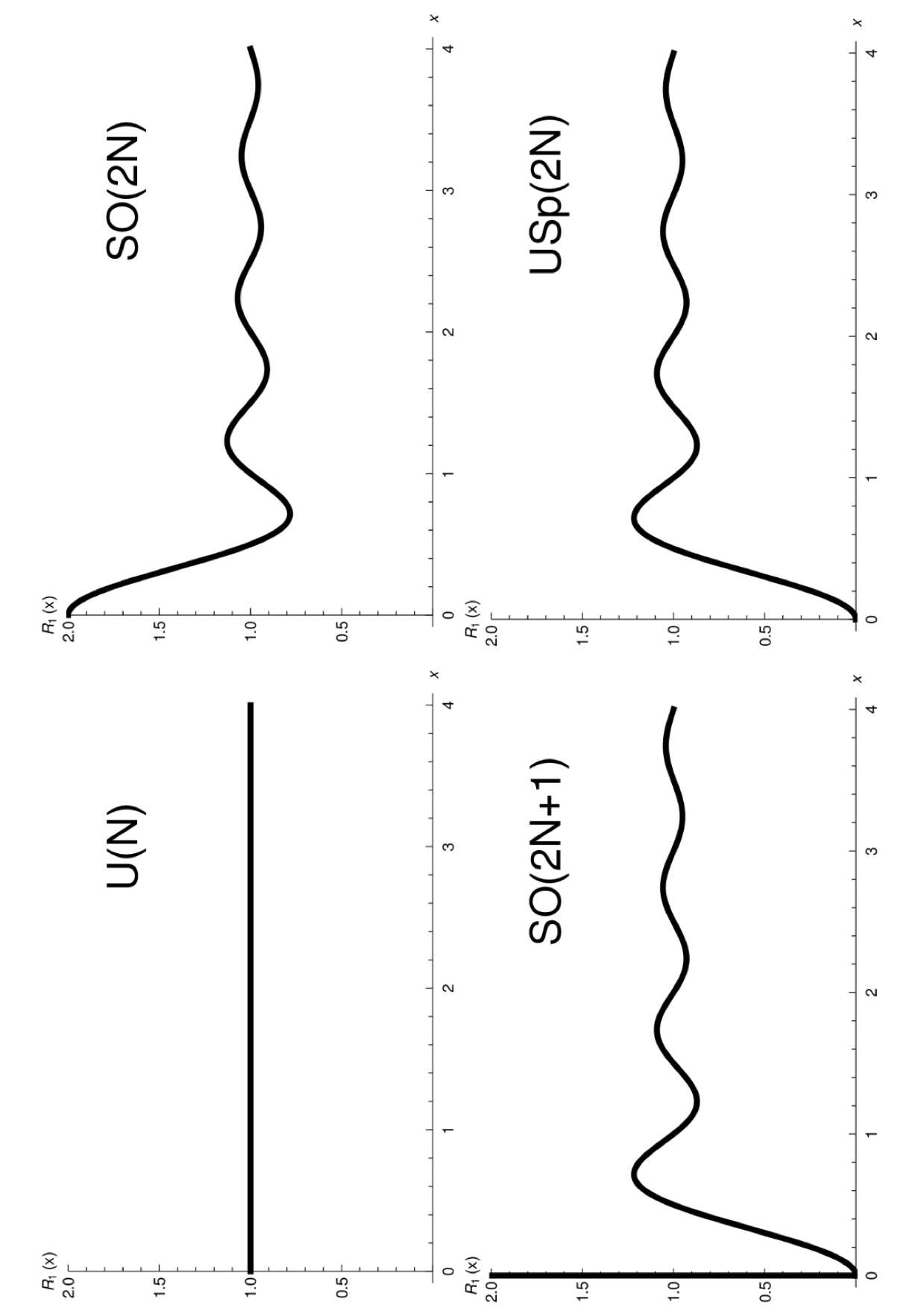}
\caption{The $x$-axis measures distance round the unit circle anti-clockwise from the point 1, in units of mean spacing. The height of the curve gives a relative likelihood of eigenvalues occuring at that position as we range over the ensemble.  The thick vertical line at the origin of the $SO(2N+1)$ plot represents a Dirac delta function. } \label{fig:onelevels}
\end{center}
\end{figure}

As $N$ becomes large and we evaluate the characteristic polynomial closer and closer to the point 1, it is in the $SO(2N+1)$ ensemble that the moment of the logarithmic derivative grows fastest with $N$; this can be seen from the factor $(N/a)^K$ in Theorem \ref{thm4}, as well as numerically when sampling as little as $100$ matrices for $N$ over 60. This is because every matrix in this ensemble has an eigenvalue equal to 1. In fact, writing the $SO(2N+1)$ logarithmic derivative out in terms of its eigenvalues ($1,e^{\pm i\theta_1}, \dots, e^{\pm i\theta_N}$), we can isolate the term that dominates as $s$ approaches 1. We have that 
\begin{align}\label{SO2N+1logder}
    \frac{\Lambda_X'}{\Lambda_X}(s)=\frac{-1}{1-s}+\sum_{n=1}^N \left( \frac{-e^{i\theta_n}}{1-se^{i\theta_n}}-\frac{e^{-i\theta_n}}{1-se^{-i\theta_n}} \right) \nonumber \\
    = \frac{-1}{1-s}+ \sum_{n=1}^N \frac{2s - 2\cos(\theta_n)}{s^2 - 2s\cos(\theta_n) + 1}.
\end{align} The leading order behaviour, $\left(\frac{-N}{a}\right)^K$ in Theorem \ref{thm4}, comes entirely from substituting $s=e^{-a/N}$ into  $\frac{-1}{1-s}$ as $N\rightarrow \infty$ and raising to the $K^{th}$ power.  This term corresponds to the eigenvalue at 1 and is always negative since $s \to 1^{-}$, therefore all odd moments in Theorem \ref{thm4} are negative as well.  The dominance of this term also comes from the fact that other eigenvalues are repelled from 1 (see Figure \ref{fig:onelevels}) and so it is unlikely there are other nearby eigenvalues making a significant contribution. 

For the even orthogonal ensemble $SO(2N)$, the logarithmic derivative can again be written in terms of the eigenvalues, which come in complex conjugate pairs ($e^{\pm i\theta_1}, \dots, e^{\pm i\theta_N}$):
\begin{align}\label{SO2Nlogder}
    \frac{\Lambda_X'}{\Lambda_X}(s)= \sum_{n=1}^N \left( \frac{-e^{i\theta_n}}{1-se^{i\theta_n}}-\frac{e^{-i\theta_n}}{1-se^{-i\theta_n}} \right) \nonumber \\
    =  \sum_{n=1}^N \frac{2s - 2\cos(\theta_n)}{s^2 - 2s\cos(\theta_n) + 1}.
\end{align}

One can consider two limits to understand the behaviour of the logarithmic derivative near 1. First, if we take the limit as $s \to 1$ with a fixed $\theta_n \neq 0$, that is, for a fixed matrix with no eigenvalue at 1, then each term in the sum 
\begin{align}\label{lim_in_s}
    \frac{2s - 2\cos(\theta_n)}{s^2 -2s\cos(\theta_n) + 1} \longrightarrow 1,
\end{align}
hence the logarithmic derivative of a matrix with no eigenvalue at 1, in this limit, is simply $N$. 

For the second limit, one can instead fix the point $s$ where we are evaluating the logarithmic derivative, and imagine the limit $\theta_n \to 0$. Then,  
\begin{align}\label{lim_in_theta}
    \frac{2s - 2\cos(\theta_n)}{s^2 -2s\cos(\theta_n) + 1} \longrightarrow \frac{2}{s-1},
\end{align}
which tells us that regardless of the ensemble, the logarithmic derivative evaluated near 1 will be dominated by large negative terms if there are eigenvalues close to 1. The eigenvalues of matrices in $SO(2N)$ show no repulsion from the point 1.  In fact, as we can see from Figure \ref{fig:onelevels}, there is a reasonable likelihood of finding an eigenvalue near the point 1. So while large contributions from eigenvalues very close to 1 are not guaranteed as they are in the odd orthogonal ensemble, they appear often enough that the first moment, as well as all other odd moments, are negative. Numerically, we see occurrences of negative values when we generate $10^2$ matrices for larger $N$, around 160, and with $10^4$ matrices for $N$ as small as 20. This behaviour is reflected in the $N^K/a^{K-1}$ factor in Theorem \ref{thm2} in that it doesn't grow as fast as the $SO(2N+1)$ case. 

Finally, the logarithmic derivative over the symplectic ensemble can be written exactly as in \eqref{SO2Nlogder}, and the two limits \eqref{lim_in_s} and \eqref{lim_in_theta} apply to this ensemble as well. However, the matrices in $USp(2N)$ have eigenvalues that show quadratic repulsion from the point 1, as seen in Figure \ref{fig:onelevels}. This means matrices with eigenvalues close to 1 are very rare in this ensemble, implying that large negative values of the logarithmic derivative are rare and smaller positive values are common. Indeed, when sampling numerically, we require at least $10^6$ matrices to begin to see some negative values. Correspondingly, we see the slowest growth in terms of $N$ in Theorem \ref{thm3}, compared to Theorems \ref{thm2} and \ref{thm4}, reflected in the smaller power of $a$ in these moments (recalling that $a = o(1)$ as $N \to \infty$). In this ensemble, the logarithmic derivative attains negative values very rarely, but when it does, the magnitude is much larger than more common values of the logarithmic derivative.  These rare but large events are therefore not enough to be seen in the average or the first few moments in the symplectic ensemble, but as one takes larger moments, their contribution is magnified; this explains why larger odd moments are eventually negative in the symplectic ensemble as well. 

\section{Preliminaries}\label{prems}
Throughout the paper, we denote the Vandermonde determinant by $\Delta(x)$, where
\begin{equation}\label{vanderdet}
    \Delta(x) = \prod_{1 \leq j < k \leq K} (x_k - x_j),
\end{equation} and it is the determinant of the Vandermonde matrix: 
\begin{align}
\begin{bmatrix}\label{vandermatrix}
1 & 1 & \dots & 1 \\
x_1 & x_2 & \dots & x_K \\
x_1^2 & x_2^2 & \dots & x_K^2 \\
\vdots & \vdots & \ddots & \vdots \\
x_1^{K-1} & x_2^{K-1} & \dots & x_K^{K-1}
\end{bmatrix}.    
\end{align} In our calculations, for various functions $f(x)$ we will come across multiple contour integrals of the following form:  
\begin{eqnarray}
   && \int_{\abs{u_i}=1} f(u_1)f(u_2) \cdots f(u_K) \Delta(u^2)\Delta(u) \mathrm{d}\mathbf{u} \nonumber \\
    &&= \int_{\abs{u_i}=1} f(u_1)f(u_2) \cdots f(u_K) \left( \sum_{\sigma \in S_K} sgn(\sigma) \prod_{i=1}^K u_i^{2\sigma(i) - 2} \right) \left( \sum_{\tau \in S_K} sgn(\tau) \prod_{k=1}^K u_{k}^{\tau(k)-1} \right) \mathrm{d}\mathbf{u},\\
   && = \sum_{\sigma \in S_K}\sum_{\tau \in S_K}\int_{\abs{u_i}=1} f(u_{\tau(1)})f(u_{\tau(2)}) \cdots f(u_{\tau(K)}) \left(  sgn(\sigma) \prod_{i=1}^K u_{\tau(i)}^{2\sigma(i) - 2} \right) \left(  sgn(\tau) \prod_{k=1}^K u_{\tau(k)}^{\tau(k)-1} \right) \mathrm{d}\mathbf{u},\nonumber
\end{eqnarray}
where in the second line we use Leibniz notation for the Vandermonde and in the third line we relabel the variables so that $u_k$ is replaced with $u_{\tau(k)}$ throughout. We now relabel the $\sigma$ sum so that $\sigma\rightarrow\sigma\tau$ and find
\begin{eqnarray}
 &&   \int_{\abs{u_i}=1} f(u_1)f(u_2) \cdots f(u_K) \Delta(u^2)\Delta(u) \mathrm{d}\mathbf{u} = \sum_{\sigma \in S_K}\sum_{\tau \in S_K}\int_{\abs{u_i}=1} f(u_{\tau(1)})f(u_{\tau(2)}) \cdots f(u_{\tau(K)}) \nonumber\\
&& \qquad\qquad\qquad \times \left(  sgn(\sigma\tau) \prod_{i=1}^K u_{\tau(i)}^{2\sigma(\tau(i)) - 2} \right) \left(  sgn(\tau) \prod_{k=1}^K u_{\tau(k)}^{\tau(k)-1} \right) \mathrm{d}\mathbf{u}.
\end{eqnarray} 
 The two $sgn(\tau)$ now cancel and we see that every one of the $K!$ terms in the $\tau$ sum gives exactly the same contribution. So 
 \begin{eqnarray}
 &&   \int_{\abs{u_i}=1} f(u_1)f(u_2) \cdots f(u_K) \Delta(u^2)\Delta(u) \mathrm{d}\mathbf{u} \nonumber \\
  &&= \int_{\abs{u_i}=1} K! f(u_1)f(u_2) \dots f(u_K) \left( \sum_{\sigma \in S_K} sgn(\sigma)  u_1^{2\sigma(1) - 2} u_2^{2\sigma(2) - 1} u_3^{2\sigma(3)} \dots u_K^{2\sigma(K) + K-3}\right) \mathrm{d}\mathbf{u} \nonumber \\
&&= \int_{\abs{u_i}=1} K! f(u_1)f(u_2) \dots f(u_K) \det\left[ u_i^{2j + i -3} \right]_{i,j = 1}^K    \mathrm{d}\mathbf{u} \nonumber \\
&&=  K!  \det\left[\int_{\abs{u_i}=1} f(u_i) u_i^{2j + i -3}\mathrm{d}u_i \right]_{i,j = 1}^K     \label{vandet}.
\end{eqnarray} 
To prove Theorems \ref{thm2}, \ref{thm3} and \ref{thm4} we will use two results, stated here and recalled in later sections as well. 
\begin{prop}{(Conrey, Forrester, Snaith \cite{CFS}, Proposition 2.3)}\label{propCFS} 
Let $z(x) := \frac{1}{1-e^{-x}}$, \, and \,  $\mathfrak{Re}(\gamma_q) \geq 0 \, \forall q$. Then, for $N \geq Q$, 
\begin{align}\label{CFSEO}
     \int_{SO(2N)} \frac{\prod_{k=1}^K \Lambda_X(e^{-\alpha_k})}{\prod_{q=1}^Q \Lambda_X (e^{-\gamma_q})}\mathrm{dX} = 
     \sum_{\varepsilon \in \{-1, 1\}^K}e^{N\sum_{k=1}^K (\varepsilon_k\alpha_k)}\frac{\prod\limits_{1 \leq j < k \leq K} z(\varepsilon_j\alpha_j + \varepsilon_k\alpha_k)\prod\limits_{1 \leq q \leq r \leq Q} z(\gamma_q + \gamma_r)}{\prod_{k=1}^K\prod_{q=1}^Q z(\varepsilon_k\alpha_k+\gamma_q) e^{N\sum_{k=1}^K \alpha_k}}.    
\end{align}   
\end{prop}

\begin{lemma}{(Conrey, Farmer, Keating, Rubinstein and Snaith \cite{CFKRS}, Lemma 2.5.2)}\label{CFKRSlem}
Consider a function $F(w) = F(w_1, \dots, w_K)$ of $K$ variables which is symmetric and regular near $(0, \dots, 0)$, and a function $f(w)$ with a simple pole of residue 1 at $w = 0$ but is otherwise analytic in $\abs{w}\leq 1$. Given 
\begin{align}\label{11}
   H(w_1, \dots, w_K) = F(w_1, \dots, w_K) \prod_{1 \leq j < k \leq K} f(w_j + w_k), 
\end{align}
or
\begin{align}\label{12}
    H(w_1, \dots, w_K) = F(w_1, \dots, w_K) \prod_{1 \leq j \leq k \leq K} f(w_j + w_k),   
\end{align}
then, for $ \lvert \alpha_k \rvert < 1 $,
\begin{align}\label{13}
   \sum_{\varepsilon \in \{ -1, 1\}^K} H(\varepsilon_1 \alpha_1, \dots, \varepsilon_K \alpha_K) = \frac{(-1)^{K(K-1)/2}2^K}{K!(2\pi i)^K} \nonumber \\ \nonumber \\
   \times \oint_{\lvert w_1 \rvert = 1} \dots \oint_{\lvert w_K \rvert = 1} \frac{H(w_1, \dots, w_K) \Delta^2(w_1^2, \dots, w_K^2)\prod_{k=1}^K w_k}{\prod_{j=1}^K \prod_{k=1}^K (w_k-\alpha_j)(w_k + \alpha_j)} \prod_{k=1}^K \text{d}w_k
\end{align}
and 
\begin{align}\label{14}
   \sum_{\varepsilon \in \{ -1, 1\}^K} \left(\prod_{j=1}^K \varepsilon_j \right) H(\varepsilon_1 \alpha_1, \dots, \varepsilon_K \alpha_K) = \frac{(-1)^{K(K-1)/2}2^K}{K!(2\pi i)^K} \nonumber \\ \nonumber \\
   \times \oint_{\lvert w_1 \rvert = 1} \dots \oint_{\lvert w_K \rvert = 1} \frac{H(w_1, \dots, w_K) \Delta^2(w_1^2, \dots, w_K^2)\prod_{k=1}^K \alpha_k}{\prod_{j=1}^K \prod_{k=1}^K (w_k-\alpha_j)(w_k + \alpha_j)} \prod_{k=1}^K \text{d}w_k.
\end{align}
\end{lemma}
\vspace{5ex}
 The first case of Lemma \ref{CFKRSlem} (\eqref{11} and \eqref{13}) applies to our main computation, for the even orthogonal ensemble. The second form (equations \eqref{12} and \eqref{13}) will be used in Section \ref{Symplectic} for the symplectic ensemble, and finally, \eqref{11} and \eqref{14} will be used for the odd orthogonal case. 
 
\section{Definition and properties of the set of matrices $\mathcal{M}$}\label{theset}
In this section we define a set of matrices $\mathcal{M}$ that recur in our moment calculations, and derive some of their properties. The general structure of this set allows us to use these properties to prove our theorems over the even orthogonal, symplectic and odd orthogonal ensembles by specifying the appropriate parameters. 
\begin{defi}\label{M}
 The matrices $M \in \mathcal{M}$ have their $(i,j)^{th}$ entry of the form:
\begin{align}
M_{i,j} =  \left[
     \begin{array}{l}
       \frac{2^{e_j}}{2\pi i} \oint \frac{u_i^{2n_j+h_i}\exp(2t /(u_i^2-1))}{(u_i-1)^{K+e_j}(u_i+1)^{e_j}} \mathrm{du_i} 
     \end{array}
   \right]_{i,j=1}^{K},
\end{align}
with $e_j$ non-negative integers and with $n_j$ and  $h_i$ being integers such that  $ h_1<h_2<\cdots <h_K$. The contours of integration enclose 1 and -1. 
\end{defi}

The integral entries of the type  
\begin{equation}\label{intdef}
    I(r, E) := \frac{2^{E}}{2\pi i} \oint \frac{u^r \exp(2t/(u^2-1))}{(u-1)^{K+E}(u+1)^E} \mathrm{du},
\end{equation}
where $r$, $E$ and $K$ are integers ($E$ and $K$ will always be non-negative), satisfy a recursion formula. One can verify directly that 
\begin{equation}\label{recursion}
    I(r, E) = 2I(r-2,E-1) + I(r-2, E).
\end{equation}

\begin{defi}[Degrees]
We define the degree of $I(r,E)$ to be
\begin{equation}
    r - K - 2E.
\end{equation}
The degree of column $J$ in a matrix  $M\in\mathcal{M}$ is the largest degree of any integral occurring in that column; we denote it $D_J(M)$.  Remark that the $K^{th}$ row always determines the maximal degree of a column: $2n_J+h_K-K-2e_J$. Also, all the column degrees of $M\in\mathcal{M}$ have the same parity.  Finally, the degree of a matrix  $M\in\mathcal{M}$ is the sum of its column degrees:
\begin{equation}
    D(M) = \sum_{j=1}^{K} D_j(M).
\end{equation}
\end{defi}

It will be useful to note the degree of each one of the three terms in the recursion formula \eqref{recursion}:
\begin{align*}
  I(r,E) & \quad \text{has degree} \quad r-K -2E \\
  I(r-2,E-1) & \quad \text{has degree} \quad r -K-2E \\
  I(r-2,E)  & \quad \text{has degree} \quad r-K -2E-2.
\end{align*}

This tells us that each time we apply the recursion \eqref{recursion} to all the entries in the $J^{th}$ column of a matrix $M \in \mathcal{M}$, we can use the property 
\begin{equation}\label{detid2}
     \det(a_1, a_2, \dots, a_i+b_i,\dots,a_n) = \det\left(a_1, a_2, \dots,a_i,\dots, a_n\right) + \det\left(a_1, a_2, \dots, b_i,\dots, a_n\right),
\end{equation}
where the $a_i$ and $b_i$ represent either the rows or columns of the matrix, to split its corresponding determinant, $\det(M)$ into a sum of two determinants of matrices in $\mathcal{M}$, of which one has the same column degree $D_J(M)$ with an extra factor of 2, and the other has lower column degree by 2. We use this fact to prove Lemma \ref{lem2} below.

\begin{lemma}\label{lem2}
    For a matrix $M \in \mathcal{M}$, suppose that two of its columns have the same degree; that is,  $D_J(M) = D_{J^{'}}(M)$ for some $J \neq J^{'}$. Then, $\exists$  matrices $M_b$ and a finite $B$ such that $\det(M) = \sum\limits_{b=1}^B\det(M_b)$, where each $M_b\in \mathcal{M}$ and $\max\limits_b D(M_b)<D(M)$. 
\end{lemma}
\begin{proof}
This lemma shows that for a given matrix in $\mathcal{M}$, if two of its columns have the same degree, then we can reduce the overall degree of the matrix. Indeed, we recall that the degree of a column is given by the maximal degree of its entries which, for matrices in $\mathcal{M}$, comes from the $K^{th}$ row. Then, if two columns have the same degree,
\begin{align*}
    D_J(M) = 2n_J+h_K-K - 2e_J = 2n_{J'} +h_K-K - 2e_{J^{'}} = D_{J^{'}}(M) \\
    \implies n_{J^{'}} = n_{J} - (e_J - e_{J^{'}}).
\end{align*}
Assume for simplicity that $e_J > e_{J^{'}}$ and let $N = e_J - e_{J^{'}}=n_J-n_{J'}$. We apply the recursion \eqref{recursion} $N$ times to each entry in column $J$, with $r_{i,J} = 2n_J + h_i$. Note that $r_{i,J}-2N=r_{i,J'}$. For each row $i$, at each step $b$ of the recursion, we obtain lower degree terms which we denote $\ell_{b,i}$. Then, the $i^{th}$ entry in column $J$ becomes:
\begin{align}\label{q1}
    I(r_{i,J}, e_J) = 2I(r_{i,J}-2, e_J - 1) + I(r_{i,J}-2, e_J) \\
    = 2I(r_{i,J}-2, e_J - 1) + \ell_{1,i} \nonumber \\
    = 4I(r_{i,J}-4, e_J-2) + 2\ell_{2,i} + \ell_{1,i} \nonumber \\
    = 8I(r_{i,J}-6, e_J-3) + 4\ell_{3,i} + 2\ell_{2,i} + \ell_{1,i} \nonumber \\
    = \dots = 2^NI(r_{i,J}-2N, e_J - N) + \sum_{b=1}^{N}2^{b-1}\ell_{b,i} \nonumber \\
    = 2^NI(r_{i,J^{'}}, e_{J^{'}}) + \sum_{b=1}^{N}2^{b-1}\ell_{b,i}. \label{rhs}
\end{align}

Now we can replace the $i^{th}$ entry in the $J^{th}$ column by  \eqref{rhs} without affecting the determinant and using (\ref{detid2}), split the determinant $\det(M)$ into a sum of determinants of matrices $\{M_b\}_{1 \leq b \leq N+1}$, where $M_{N+1}$ has the term $2^NI(r_{i,J^{'}}, E_{J^{'}})$ in its $J^{th}$ column and $i^{th}$ row, $M_{N}$ has $2^{N-1}\ell_{N,i}$ in its $J^{th}$ column and $i^{th}$ row, and generally, for $ 1 \leq b \leq N$, $M_b$ has the term $2^{b-1}\ell_{b,i}$ in its $J^{th}$ column and $i^{th}$ row. Then $M_{N+1}$ has two columns, $J$ and $J^{'}$ which are scalar multiples of each other and its determinant vanishes. For $ 1 \leq b \leq N$, $M_b$ has lower matrix degree, since $D_J(M_b) < D_J(M)$ and $D_L(M_b) = D_L(M)$ for all other columns $L \neq J$.
\end{proof}

\begin{rem} \label{rem-parity} It is also useful to remember that all the column degrees of a given $M \in \mathcal{M}$ have the same parity and, due to the recursion \eqref{recursion} always reducing the degree of an integral by 2, all of the matrices $M_b$ in the proof of lemma \ref{lem2} similarly have all column degrees of the same parity as those of $M$.
\end{rem}

The next propositions also concern matrices in the set $\mathcal{M}$; we establish a minimum matrix degree for $M \in \mathcal{M}$ to have a non-zero determinant. We then relate the matrix degree to column derivatives.

\begin{prop} \label{mindegree} The determinant of $M \in \mathcal{M}$ is zero if the minimum column degree is less than or equal to -2.  That is if 
\begin{equation}
    \min_{1\leq j\leq K}(2n_j+h_K-K-2e_j)\leq -2.
\end{equation}
\end{prop}
\begin{proof} This is proved by considering the integrals in the column where each element has degree less than or equal to -2. In an integral of the type \eqref{intdef} the radius of the contour of integration  can be increased without crossing any poles, and thus without changing the integral. Because the degree of the integral is -2 or less, the integrand shrinks faster than the length of the contour grows, leading to a column of zeroes, and this causing a vanishing determinant. 
\end{proof}

\begin{prop}\label{odddegree}
 For a matrix $M \in \mathcal{M}$ as defined in Definition \ref{M}, if the parity of all column degrees is odd and if the degree of the matrix $D(M) < K(K-2)$ then $\det(M) = 0$. Furthermore, if $\det(M)\neq 0$ and $D(M) = K(K-2)$, it follows that the columns of $M$ must have degrees, in some order,
$$-1, 1, 3, 5, \dots, 2K-5, 2K-3.$$ 
\end{prop}
\begin{proof}
By Lemma \ref{lem2}, we know that if a matrix $M \in \mathcal{M}$ has two columns with equal column degrees, then we can split the determinant $\det(M)$ into a sum of determinants of matrices in $\mathcal{M}$ of lower degree (still with odd column degrees, by Remark \ref{rem-parity}). Therefore, we now assume without loss of generality that we are working with matrices $M \in \mathcal{M}$ whose column degrees are all odd and pairwise distinct. If a column has degree $D_J(M) < -1$ then the column vanishes (by Proposition \ref{mindegree}) and $\det(M) = 0$. Therefore, $-1$ is the minimum column degree for a matrix $M \in \mathcal{M}$ with non-vanishing determinant. Since the degree of a column must be odd, the minimal matrix degree is simply
\begin{align}
    -1+ 1+ 3 + \dots + 2K-3 = \sum_{j=0}^{K-1} (2j-1) = K(K-2).
\end{align}
\end{proof}

\begin{prop}\label{evendegree}
 For a matrix $M \in \mathcal{M}$ as defined in Definition \ref{M}, if the parity of all column degrees is even and if the degree of the matrix $D(M) < K(K-1)$ then $\det(M) = 0$. Furthermore, if $\det(M)\neq 0$ and $D(M) = K(K-1)$, it follows that the columns of $M$ must have degrees, in some order,
$$0, 2, 4, 6, \dots, 2K-4, 2K-2.$$ 
\end{prop}
\begin{proof}
The proof follows exactly as that of Proposition \ref{odddegree}. Here the column degrees are even, and $0$ is the minimum column degree for a matrix $M \in \mathcal{M}$ with non-vanishing determinant. The minimal matrix degree is simply
\begin{align}
    0+ 2+ 4 + \dots + 2K-2 = \sum_{j=0}^{K-1} 2j = K(K-1).
\end{align}
\end{proof}

\begin{prop} \label{maxderivatives} If we take $d$ derivatives with respect to $t$ of the determinant of a matrix $M \in \mathcal{M}$, we obtain a sum of determinants, each of which is from a matrix in $\mathcal{M}$, and has column degrees of the same parity as $M$. Each determinant in the sum will be zero if the matrix degree $D(M)<K(K-2)+2d$, when the column degrees of $M$ have odd parity, or if $D(M)<K(K-1)+2d$, when the column degrees of $M$ have even parity. 
\end{prop}

\begin{proof} We recall that the determinant is multilinear and its derivative can be written as a sum:
\begin{equation}\label{detid}
    \frac{\mathrm{d}}{\mathrm{d}t} \det(a_1, a_2, \dots, a_n) = \det\left(\frac{\mathrm{d}}{\mathrm{d}t}a_1, a_2, \dots, a_n\right) + \dots + \det\left(a_1, a_2, \dots, \frac{\mathrm{d}}{\mathrm{d}t} a_n\right),
\end{equation}
where the $a_i$ represent either the rows or columns of the matrix. In our case, we will be working with the columns (indexed by $j$). The derivative of an element of a matrix $M\in \mathcal{M}$ is
\begin{equation}
     \frac{\mathrm{d}}{\mathrm{d}t}  \frac{2^{e_j}}{2\pi i} \oint \frac{u_i^{2n_j+h_i}\exp(2t /(u_i^2-1))}{(u_i-1)^{K+e_j}(u_i+1)^{e_j}} \mathrm{du_i} = \frac{2^{e_j+1}}{2\pi i} \oint \frac{u_i^{2n_j+h_i}\exp(2t /(u_i^2-1))}{(u_i-1)^{K+e_j+1}(u_i+1)^{e_j+1}} \mathrm{du_i} , 
\end{equation} which, recalling Definition \ref{intdef}, we can write simply as
\begin{equation}
    \frac{\mathrm{d}}{\mathrm{d}t} I(2n_j+h_i, e_j)=I(2n_j+h_i, e_j+1).
\end{equation} Differentiating every element in a column gives us another matrix in $\mathcal{M}$ and this new matrix has degree two less than the original $M$.  Applying \eqref{detid} $d$ times, Proposition \ref{maxderivatives} follows from Propositions \ref{odddegree} and \ref{evendegree} and the fact that every time a column is differentiated, the degree of the matrix decreases by 2. 
\end{proof}
The following proposition shows why the converse of Proposition \ref{maxderivatives} does not hold, by characterising matrices whose degrees are greater than $K(K-2) + 2d$ (or $K(K-1) + 2d$) yet their determinants vanish if multiple columns are differentiated. This is due to differentiation causing two columns to become scalar multiples of each other in these particular cases, which will occur in Section \ref{Symplectic}. For this proposition we need to define the secondary column and matrix degrees. 
\begin{defi}[Secondary degrees]
For a matrix $M\in \mathcal{M}$, the secondary column degree, denoted $\Tilde{D}_j(M)$ is the second largest degree in the column $j$, which occurs in the $K-1^{th}$ row. Similarly, the secondary matrix degree is defined as the sum of the secondary column degrees: 
\begin{align}
\Tilde{D}(M) = \sum_{j=1}^K \Tilde{D}_j(M).
\end{align}
\end{defi}

\begin{prop} \label{mysterylemma}
For a matrix $M \in \mathcal{M}$ as defined in \ref{M} with $n_J = J$ for all $J$ and $e_1=\dots=e_K=0$, if $D_1(M) < 2K-1$ and $\Tilde{D}(M) < K(K-2)$ then $\frac{\mathrm{d^K}}{\mathrm{dt^K}}\det(M) = 0$. 
\end{prop}
\begin{proof}
By hypothesis, all column degrees are pairwise distinct. We assume that $D_J(M) \geq 2J - 1$ or $D_J(M) \geq 2J$ for matrices with odd and even parity column degrees respectively, since otherwise we apply Proposition \ref{maxderivatives} and find the determinant vanishes. If the secondary column degrees are $\Tilde{D}_J(M) \leq 2J-4$ and $\det(M) \neq 0$, the first column only has one nonzero entry, in its last row. This means that the first column can be differentiated until its degree is reduced to $-1$ before causing a vanishing determinant (or to 0 for even column degrees). For example, if $D_1(M) = 1$, the first column survives one derivative, and if $D_1(M) = 3$ it survives two derivatives. If $D_1(M)<2K-1$, as in the condition in the proposition, then the first column will not survive $K$ derivatives and to maintain a non-zero determinant we will be forced to differentiate the second column.  Once we differentiate the second column, all of its entries except for the last will have degree $< -1$ (or less than $0$ for even degrees) since the secondary degree is also reduced by 2 after differentiation. Therefore, both the first and second columns would be columns of zeroes except for their final row, which means they are scalar multiples of each other. Had we differentiated any other column $J^{'}$, then it would become a scalar multiple of the $J^{'} - 1 ^{th}$ column and cause a vanishing determinant as well. Thus, $\Tilde{D}_J(M) \geq 2J -3$, and the Proposition follows since
\begin{align}
    \sum_{j=1}^K (2j-3) = K(K-2).
\end{align}
\end{proof}

The following lemma that helps us compute derivatives of matrices of type $\mathcal{M}$.
\begin{lemma} \label{multiplicitylemma}
Let $M\in \mathcal{M}$, with $n_j=j$ for $j=1,\ldots,K$ and with $e_1=\dots=e_K=0$. Assume that either $h_K=K-1$, in which case the column degrees have odd parity and the matrix degree is $D(M)=K(K-2)+2K$, or else $h_K=K$, in which case the column degrees have even parity and the matrix degree is $D(M)=K(K-1)+2K$. The $K^{th}$ derivative of $\det M$ with respect to $t$, evaluated at $t=0$, is:
\begin{equation}\label{multiplicity}
    \frac{\mathrm{d}^{K}}{\mathrm{d}t^{K}}\det \left[
     \begin{array}{l}
       \frac{1}{2\pi i} \oint \frac{u_i^{2j+h_i}\exp(2t /(u_i^2-1))}{(u_i-1)^{K}} \mathrm{du_i}   \end{array}
   \right]_{i,j=1}^{K}\Bigg|_{t=0}=2^K \det \left[
     \begin{array}{l}
       \binom{2j+h_i-2}{K-1} \end{array}
   \right]_{i,j=1}^{K},
\end{equation}
applying the convention that $\binom{2j+h_i-2}{K-1}=0$ if $2j+h_i-2 \leq K-2$.
\end{lemma}
\begin{proof}
By \eqref{detid} the derivative is the sum of determinants of the form
\begin{equation}\label{mat1}
   \det \left[
     \begin{array}{l}
       \frac{2^{E_j}}{2\pi i} \oint \frac{u_i^{2j+h_i}\exp(2t /(u_i^2-1))}{(u_i-1)^{K+E_j}(u_i+1)^{E_j}} \mathrm{du_i}   \end{array}
   \right]_{i,j=1}^{K},
\end{equation}
where $E_j$ is the number of times the $j^{th}$ column has been differentiated, and $\sum_{j=1}^K E_j=K$.

We recall that the $(i,j)^{th}$ element in \eqref{mat1} is the integral $ I(2j+h_i, E_j)$. We use the recursion formula $\eqref{recursion}$, as we did in Lemma \ref{lem2}, to reduce $E_j$ to $0$.
\begin{align}
    I(2j+h_i, E_j) = 2I(2j+h_i-2, E_j - 1) + \ell_1 \nonumber \\
    = 4I(2j+h_i-4, E_j -2) + \ell_2 + \ell_1= \dots \nonumber \\
    = 2^{E_j}I(2j+h_i-2E_j, 0) + \mathcal{L} \nonumber \\
    = \frac{2^{E_j}}{2\pi i} \oint \frac{u_i^{2j+h_i-2E_j}\exp(2t/(u_i^2-1))}{(u_i-1)^K}\mathrm{du_i} + \mathcal{L}. \label{intpluslower}
\end{align}
Where $\mathcal{L} := \sum\limits_m \ell_m$. 

We repeat this process for every element in each column $J$. Note that for each column, we can split the determinant as a sum of determinants of matrices in $\mathcal{M}$ using \eqref{detid2}; one with the integral from \eqref{intpluslower} in the $J^{th}$ column, the others with the lower degree terms of $\mathcal{L}$. Since the matrix $M$ has degree $D(M)=K(K-2)+2K$ if the column degrees have odd parity (resp. matrix degree $D(M)=K(K-1)+2K$ if the column degrees have even parity), after $K$ derivatives the matrices of the form \eqref{mat1} have degree $K(K-2)$ (resp. $K(K-1)$). By Proposition \ref{odddegree} and Proposition \ref{evendegree} the matrices with lower degree terms from $\mathcal{L}$ in one or more columns have too small a matrix degree and therefore their determinant is 0.  Thus the derivative in \eqref{multiplicity} becomes the sum of determinants of the form
\begin{equation}\label{mat2}
   \det \left[
     \begin{array}{l}
       \frac{2^{E_j}}{2\pi i} \oint \frac{u_i^{2j+h_i-2E_j}\exp(2t /(u_i^2-1))}{(u_i-1)^{K}} \mathrm{du_i}   \end{array}
   \right]_{i,j=1}^{K},
\end{equation}
where $\sum_{j=1}^K E_j=K$.

The matrices in \eqref{mat2} belong to $\mathcal{M}$, so we recall that, by Proposition \ref{odddegree} and Proposition \ref{evendegree}, the determinant is non-zero only if the column degrees take, in some order, the values 
$-1, 1, 3, \dots, 2K-5, 2K-3$ (resp. $0,2,4,\dots,2K-4,2K-2$). One way that this can happen is if each column is differentiated exactly once (although we will see later that this is not the only scenario).

In the specific case where $E_j=1$ for $j=1,\ldots,K$, we set $t = 0$  and define
\begin{equation}\label{mhat}
 \det\hat{M}:=  \det \left[
     \begin{array}{l}
       \frac{2}{2\pi i} \oint \frac{u_i^{2j+h_i-2}}{(u_i-1)^{K}} \mathrm{du_i}   \end{array}
   \right]_{i,j=1}^{K}.
\end{equation}
We compute the integral as a residue, obtaining
\begin{align}
    \frac{1}{2\pi i} \oint \frac{u_i^{2j+h_i-2}\mathrm{du_i} }{(u_i-1)^K}     = \binom{2j+h_i-2}{K-1}. 
\end{align}
Then, if $E_j=1$ for all $j$, at $t=0$, \eqref{mat2} becomes 
\begin{equation}\label{binomialform}
  \det \hat{M}=  2^K \det \left[
     \begin{array}{l}
       \binom{2j+h_i-2}{K-1} \end{array}
   \right]_{i,j=1}^{K},
\end{equation}
 applying the convention that $\binom{2j+h_i-2}{K-1}=0$ if $2j+h_i-2 \leq K-2$.  

This determinant is only one of the terms that appear in the final sum. However, the column degrees in \eqref{mat2} can only take the values $-1, 1, 3, \dots, 2K-5, 2K-3$ (resp. $0,2,4,\dots,2K-4,2K-2$) in some order, and once the degree of a column of a matrix like \eqref{mat2} is fixed, this determines the value of $2j-2E_j$ and that fixes all elements in column $j$. Thus any other matrix of the form \eqref{mat2} with non-vanishing determinant will be a permutation of the columns of the case when $E_j=1$ for $j=1,\ldots,K$ and so will have the same determinant as in \eqref{binomialform} up to sign change. To compute the total sum, we need to consider the sign and multiplicity of each determinant. Recalling that the column degrees are given by the number of times $E_j$ that the $j^{th}$ column has been differentiated, we can assign to each matrix its corresponding vector $(E_1, E_2, \dots, E_K)$. The multiplicity of its corresponding determinant is then given by the multinomial coefficient
$$\frac{K!}{\prod\limits_j^K E_j!}.$$

Furthermore, each vector $(E_1, E_2, \dots, E_K)$ corresponds to a permutation $\sigma$ in the symmetric group $S_K$. The permutation $\sigma$ encodes the ordering of the column degrees $-1,1,3,\ldots,2K-3$ (resp. $0,2,4,\dots,2K-4,2K-2$).  This is illustrated in Table \ref{tab:permutation}, where we take for example $K = 4$ and initial column degrees $1, 3, 5, 7$. We show only the vectors and permutations that correspond to non-vanishing determinants, so we insist that the column degrees are unique, otherwise the determinant would be zero by Lemma \ref{lem2}. We also note that the sign of the permutation is precisely the sign of the corresponding determinant, since it determines the number of column swaps we would need to go from the matrix in \eqref{binomialform} to any other matrix from \eqref{mat2} with non-vanishing determinant.  

\begin{table}[htpb]
\begin{center}
\begin{tabular}{|c|c|c|c|c|}
     \hline 
     $D_J$ &$ \sigma$ & $E_J$ & $\mathrm{sgn}(\sigma)$ & \text{multiplicity}  \\ 
     \hline 
    (1, 3, 5, -1) & (2, 3, 4, 1) & (0,0,0,4) & - & 1\\
    (1, 3, -1, 5) & (2, 3, 1, 4) & (0,0,3,1) & + & 4 \\
    (1, -1, 5, 3) & (2, 1, 4, 3) & (0,2,0,2) & + & 6 \\
    (1, -1, 3, 5) & (2, 1, 3, 4) & (0,2,1,1) & - & 12 \\
    (-1, 3, 5, 1) & (1, 3, 4, 2) & (1,0,0,3) & + & 4
    \\ (-1, 3, 1, 5) & (1, 3, 2, 4) & (1,0,2,1) & - & 12 \\ 
    (-1, 1, 5, 3) & (1, 2, 4, 3) & (1,1,0,2) & - & 12
    \\ (-1, 1, 3, 5) & (1, 2, 3, 4) & (1,1,1,1) & + & 24 \\
    \hline
    \end{tabular}
    \end{center}
    \caption{\label{tab:permutation} This table shows, for a manageable example with $K=4$,  the relation between: the degree of the $J^{th}$ column, $D_J$, after differentiation; the permutation, $\sigma$, describing the order of the column degrees; $E_J$, the number of times the $J^{th}$ column was differentiated; and the multiplicity, the number of ways we can achieve this set of $D_J$'s by differentiating the columns in a different order.}
\end{table} 
We see that in the permutation column of the above table, there is the constraint that $\sigma_j \leq j+1$ because matrix $M$ in this example had initial column degrees $D_1=1,D_2=3,D_3=5,D_4=7$ and differentiating always lowers a column degree by 2, so column 1 can only have degree 1 or -1 after differentiation (without the determinant becoming zero by Proposition \ref{mindegree}). 

Noting that $E_j = j+1-\sigma_j$, we now have all the elements to compute the final sum; we know the sign and multiplicity of each determinant, and can take the full sum over non-vanishing determinants. Thus,
\begin{align}\label{diagdet} 
\frac{\mathrm{d}^K}{\mathrm{dt}^K}\det(M)\Bigg|_{t=0} =     \det(\hat{M}) K! \sum_{\substack{\sigma \in S_K \\ 
    \sigma_j \leq j+1}} \frac{\mathrm{sgn}(\sigma)}{\prod\limits_j^K(j+1-\sigma_j)!}.
\end{align} 

The sum in \eqref{diagdet} can also be seen as the determinant of a Toeplitz matrix with entries $1/(j+1-i)!$ if $i \leq j+1$ and $0$ otherwise. Taking again $K = 4$ as an example, this matrix would be:
\begin{align}
   T =  \begin{bmatrix} 
    1 & 1/2! & 1/3! & 1/4! \\
    1 & 1 & 1/2! & 1/3! \\
    0 & 1 & 1 & 1/2! \\
    0 & 0 & 1 & 1 \\
    \end{bmatrix}.
\end{align}

To compute this Toeplitz determinant, we note that $T$ can easily be turned into an upper triangular matrix by first subtracting the first row from the second:
\begin{align}
   \begin{bmatrix} 
    1 & 1/2! & 1/3! & 1/4! \\
    0 & 1 -1/2! & 1/2! - 1/3! & 1/3! - 1/4! \\
    0 & 1 & 1 & 1/2! \\
    0 & 0 & 1 & 1 \\
    \end{bmatrix}
    =    \begin{bmatrix} 
    1 & 1/2 & 1/6 & 1/24 \\
    0 & 1/2 &  1/3 & 1/8 \\
    0 & 1 & 1 & 1/2! \\
    0 & 0 & 1 & 1 \\
    \end{bmatrix},
\end{align}
then subtracting 2 times the second row from the third to get

\begin{align}
    \begin{bmatrix} 
    1 & 1/2! & 1/3! & 1/4! \\
    0 & 1/2 & 1/3 & 1/8 \\
    0 & 0 & 1/3 & 1/4 \\
    0 & 0 & 1 & 1 \\
    \end{bmatrix},
\end{align}
and finally subtracting 3 times the third row from the fourth: 
\begin{align}
    \begin{bmatrix} 
    1 & 1/2! & 1/3! & 1/4! \\
    0 & 1/2 & 1/3 & 1/8 \\
    0 & 0 & 1/3 & 1/4 \\
    0 & 0 & 0 & 1/4 \\
    \end{bmatrix}.
\end{align}

This process extends inductively for arbitrary $K$ to obtain an upper triangular matrix and so we can compute the determinant by taking the product of the diagonal entries $1/i$. More generally, if we define the Toeplitz matrix $T$ as:

\begin{align}
    T = \left[ \begin{array}{ll}
      \frac{1}{(j+1-i)!}   & 1 \leq i \leq j+1 \leq K \\
       0  & \text{otherwise}
    \end{array} \right], \\
    \det(T) = \sum_{\substack{\sigma \in S_K \\ 
    \sigma_j \leq j+1}} \frac{\mathrm{sgn}(\sigma)}{\prod\limits_j^K(j+1-\sigma_j)!} = \frac{1}{K!},
\end{align}
which implies that Equation $\eqref{diagdet} = \det(\hat{M})$. This concludes our computation and the proof of Lemma \ref{multiplicitylemma}.
\end{proof}

\section{The Even orthogonal case}\label{s2} 
To begin the proof of Theorem \ref{thm2}, we will apply Lemma \ref{CFKRSlem} with functions:
\begin{align}
   F(w) = \frac{e^{N\sum_{k=1}^K w_k}}{\prod_{k=1}^K \prod_{q=1}^Q z(w_k + \gamma_q)} 
\end{align} 
and 
\begin{align} 
f(w) = z(w_j + w_k).
\end{align}

Now we can rewrite \eqref{CFSEO} in Proposition \ref{propCFS} by applying Lemma \ref{CFKRSlem}, which gives us
\begin{eqnarray}
 &&   \int_{SO(2N)} \frac{\prod_{k=1}^K \Lambda_X(e^{-\alpha_k})}{\prod_{q=1}^Q \Lambda_X (e^{-\gamma_q})} \mathrm{dX} \nonumber \\
 &&\qquad= \sum_{\varepsilon \in \{-1, 1\}^K}e^{N\sum_{k=1}^K (\varepsilon_k\alpha_k - \alpha_k)}\frac{\prod\limits_{1 \leq j < k \leq K} z(\varepsilon_j\alpha_j + \varepsilon_k\alpha_k) \prod\limits_{1 \leq q \leq r \leq Q} z(\gamma_q + \gamma_r)}{\prod_{k=1}^K\prod_{q=1}^Q z(\varepsilon_k\alpha_k+\gamma_q)}\nonumber \\
  &&  \qquad= \frac{(-1)^{K(K-1)/2}2^K}{K!(2\pi i)^K} \oint \dots \oint\frac{\prod\limits_{1 \leq j < k \leq K}z(w_j+w_k)\Delta^2(w_1^2, \dots, w_K^2) \prod_{k=1}^K w_k}{\prod_{k=1}^K \prod_{q=1}^Q z(w_k+\gamma_q) \prod_{j=1}^K \prod_{k=1}^K (w_k-\alpha_j)(w_k + \alpha_j)} \nonumber \\
  && \qquad\qquad\qquad  \times e^{N(\sum_{k=1}^K w_k - \alpha_k)} \prod\limits_{1 \leq q \leq r \leq Q} z(\gamma_q + \gamma_r) \prod_{j=1}^K \text{d}w_j , \label{eq11}
\end{eqnarray}
where the contours of integration contain the poles $\pm \alpha_j$. Next we set $K=Q$ and differentiate both sides with respect to all $\alpha_j$. Since 
\begin{align*}
   \frac{\mathrm{d}}{\mathrm{d}\alpha} \Lambda_X(e^{-\alpha}) = -\Lambda^{'}_X(e^{-\alpha})e^{-\alpha}, 
\end{align*}
\begin{align}  \prod_{j=1}^K \frac{\mathrm{d}}{\mathrm{d}\alpha_j}\int_{SO(2N)} \frac{\Lambda_X(e^{-\alpha_j})}{\Lambda_X(e^{-\gamma_j})}\mathrm{dX} = (-1)^K \int_{SO(2N)} \prod_{j=1}^K\frac{ \Lambda^{'}_X(e^{-\alpha_j})e^{-\alpha_j}}{\Lambda_X (e^{-\gamma_j})} \mathrm{dX}.
\end{align}
Similarly, since
\begin{align*}
    \frac{\mathrm{d}}{\mathrm{d}\alpha} \frac{e^{N(w-\alpha)}}{\prod_{k=1}^K (w_k-\alpha)(w_k+\alpha)} = \frac{e^{N(w-\alpha)}}{\prod_{k=1}^K (w_k-\alpha)(w_k+\alpha)}\left[-N+ \sum_{k=1}^K \frac{2\alpha}{w_k^2 - \alpha^2}\right],
\end{align*}
\begin{align}\label{question}
    \prod_{j=1}^K \frac{\mathrm{d}}{\mathrm{d}\alpha_j} \frac{e^{N\sum_{k=1}^K(w_k-\alpha_k)}}{ \prod_{k=1}^K (w_k-\alpha_j)(w_k + \alpha_j)} = \prod_{j=1}^K \frac{e^{N\sum_{k=1}^K(w_k-\alpha_k)}}{ \prod_{k=1}^K (w_k-\alpha_j)(w_k + \alpha_j)}  \left[ -N + \sum_{k=1}^K \frac{2\alpha_j}{w_k^2-\alpha_j^2}\right].
\end{align}

Next, we fix all $\alpha_j, \gamma_j = \alpha$ for all $j$, and for simplicity denote $\prod_{j=1}^K \mathrm{d}w_j = \mathrm{d}\mathbf{w}$. Then,
\begin{align}
 (-1)^K \int_{SO(2N)} \left( \frac{ \Lambda^{'}_X}{\Lambda_X }(e^{-\alpha}) \right)^K e^{-K\alpha} \, \mathrm{dX} \\ \nonumber \\
 =  \frac{(-1)^{K(K-1)/2}2^K}{K!(2\pi i)^K} \oint \dots \oint \frac{ \prod\limits_{1 \leq j < k \leq K}z(w_j+w_k)\Delta^2(w^2) \prod_{k=1}^K w_k}{ \prod_{k=1}^K z(w_k+\alpha)^K}  \prod_{1 \leq j \leq k \leq K} z(2\alpha)  \nonumber \\ \nonumber \\
 \times \quad \frac{e^{N\sum_{k=1}^K w_k}e^{-NK\alpha}}{\prod_{j=1}^K (w_j-\alpha)^K(w_j+\alpha)^K} \left[-N + \sum_{j=1}^K \frac{2\alpha}{w_j^2-\alpha^2}\right]^K \mathrm{d}\mathbf{w} \\ \nonumber \\
 = \frac{(-1)^{K(K-1)/2}2^K}{K!(2\pi i)^K} z(2\alpha)^{K(K+1)/2} e^{-NK\alpha} \oint \dots \oint \frac{ \prod\limits_{1 \leq j < k \leq K}z(w_j+w_k)\Delta^2(w^2) \prod_{k=1}^K w_k}{ \prod_{k=1}^K z(w_k+\alpha)^K}   \nonumber \\ \nonumber \\
 \times \quad \frac{e^{N\sum_{k=1}^K w_k}}{\prod_{j=1}^K (w_j^2-\alpha^2)^K} \left[-N + \sum_{j=1}^K \frac{2\alpha}{w_j^2-\alpha^2}\right]^K  \text{d}\mathbf{w}. \label{eq20}
\end{align}
Now we are ready to compute the asymptotics in $N$. First we scale our variables by $N$ by setting $\alpha = a/N$ and $w_j = au_j/N$ where $a = o(1)$ as $N \to \infty$. We are evaluating the function $z(x) = \frac{1}{x} + \frac{1}{2} + \frac{x}{12} + \mathcal{O}(x^3)$ at points $au/N$ which are getting small as $N \to \infty$, therefore we can approximate $z(x)$ by $1/x$. Putting all this together, we can rewrite \eqref{eq20}: 
\begin{eqnarray}\label{zapp}
 &&   (-1)^K \int_{SO(2N)} \left( \frac{ \Lambda^{'}_X}{\Lambda_X }(e^{-\alpha}) \right)^K e^{-K \alpha} \mathrm{dX}= \frac{(-1)^{K(K-1)/2}2^K}{K!(2\pi i)^K}\left(\frac{N}{2a}\right)^{K(K+1)/2}\nonumber\\
 &&   \qquad \times e^{-Ka} \oint \dots \oint \prod_{1 \leq j < k \leq K} \frac{N}{a(u_j+u_k)} \Delta^2\left( \frac{a^2u_j^2}{N^2}\right) \left(\frac{a}{N}\right)^K \prod_{k=1}^K u_k \times \left(\frac{a}{N}\right)^{K^2}   \prod_{k=1}^K (u_k+1)^K \nonumber  \\
 &&   \qquad\times \left(\frac{N}{a}\right)^{2K^2} \frac{e^{a\sum_{k=1}^K u_k}}{\prod_{j=1}^K (u_j^2-1)^K}\left[-N + \sum_{j=1}^K \frac{2N}{a(u_j^2-1)} \right]^K \left(\frac{a}{N}\right)^{K} \mathrm{d}\mathbf{u} \, \left(1+ \mathcal{O}\left( \tfrac{a}{N} \right)\right), \label{eq23}
\end{eqnarray}
where now the contours contain $\pm 1$.  Gathering terms and simplifying, \eqref{eq23} equals
\begin{align}
    \left(\frac{-1}{2}\right)^{\binom{K}{2}} \left(\frac{N}{a}\right)^{K}\frac{e^{-Ka}}{K!(2\pi i)^K}  \oint \dots \oint \prod_{1 \leq j < k \leq K} \frac{(u_k-u_j)^2(u_k+u_j)^2}{(u_j+u_k)} \, e^{a\sum_{k=1}^K u_k} \nonumber \\
    \times  \prod_{k=1}^K  \frac{u_k  (u_k+1)^K}{(u_k-1)^K(u_k+1)^K}\left[-a + \sum_{j=1}^K \frac{2}{(u_j^2-1)} \right]^K \mathrm{d}\mathbf{u} \, \left(1+ \mathcal{O}\left( \tfrac{a}{N} \right)\right),
\end{align}
\begin{align}
    = \left(\frac{-1}{2}\right)^{\binom{K}{2}} \left(\frac{N}{a}\right)^{K}\frac{e^{-Ka}}{K!(2\pi i)^K}  \oint \dots \oint  \Delta(u^2) \Delta(u) \, e^{a\sum_{k=1}^K u_k} \nonumber \\
    \times  \prod_{k=1}^K  \frac{u_k }{(u_k-1)^K}\left[-a + \sum_{j=1}^K \frac{2}{(u_j^2-1)} \right]^K \mathrm{d}\mathbf{u} \, \left(1+ \mathcal{O}\left( \tfrac{a}{N} \right)\right). \label{eq27}
\end{align}
We will further simplify \eqref{eq27} by factorising it completely. First we note that the term in the square brackets can be written as a $K^{th}$ derivative:
\begin{equation}\label{27}
   \left[-a + \sum_{j=1}^K \frac{2}{(u_j^2-1)} \right]^K =  \frac{\mathrm{d}^K}{\mathrm{dt}^K}\Bigg|_{t=0} \quad \exp{\left(-at + t\sum_{j=1}^K \frac{2}{(u_j^2-1)}\right)}.
\end{equation}
Recalling the symmetry of the Vandermonde determinants, we rewrite the product of the two Vandermondes $\Delta(u^2)\Delta(u)$ as a determinant under the integral, as explained in \eqref{vandet}. Our computation is now completely factorised, so we bring all the terms in \eqref{eq27} into the determinant of \eqref{vandet}, recalling that each term of the product can be brought into the determinant by multiplying one column or row by that term. Now, we have shown that:
\begin{eqnarray}\label{det}
  && (-1)^K \int_{SO(2N)} \left( \frac{ \Lambda^{'}_X}{\Lambda_X }(e^{-\alpha}) \right)^K e^{-K\alpha} \, \mathrm{dX} =  \left(\frac{-1}{2}\right)^{\binom{K}{2}} \left(\frac{N}{a}\right)^{K} \\
 && \qquad \times e^{-Ka} \frac{\mathrm{d}^K}{\mathrm{dt}^K}\Bigg|_{t=0} e^{-at}  \, \det\left[ \frac{1}{2\pi i} \oint \frac{u_i^{2j+i-2}\exp(au_i+2t/(u_i^2-1))}{(u_i-1)^K} \mathrm{du_i} \right]_{i,j = 1}^K \, \left(1+ \mathcal{O}\left( \tfrac{a}{N} \right)\right). \nonumber
\end{eqnarray}

We recall that $a$ is getting small in the asymptotic regime as $N$ tends to infinity and so we now start to examine the main term in (\ref{det}) as an expansion in small $a$. Our first step is to show that the leading order term of the determinant in \eqref{det}, with the approximation $\exp(au_i) \sim 1$, is independent of $t$ by showing that its derivative with respect to $t$ is 0. 
\begin{lemma}\label{indep}
    The determinant 
    \begin{align*}
      \det B:=  \det\left[ \frac{1}{2\pi i} \oint \frac{u_i^{2j+i-2}\exp(2t /(u_i^2-1))}{(u_i-1)^K} \mathrm{du_i} \right]_{i,j = 1}^K 
    \end{align*} is independent of $t$, and is equal to 
    \begin{align*}
     \det\left[ \binom{2j+i-2}{K-1} \right]_{i,j = 1}^K. 
    \end{align*}
\end{lemma}
\begin{proof}
We will show independence of $t$ by demonstrating that when we differentiate $\det B$ we get a sum of determinants (as on the right hand side of the identity \eqref{detid}), which are all equal to 0. Once we have independence of $t$, we fix $t = 0$ to simplify the matrix and then evaluate its determinant directly. To carry out these steps, we begin by considering the matrix elements, which are integrals:
\begin{equation}\label{integral}
    \frac{1}{2\pi i} \oint \frac{u_i^{2j+i-2}\exp(2t/(u_i^2-1))}{(u_i-1)^K} \mathrm{du_i}.
\end{equation}

Then, to compute the derivative of each column, we compute the derivative of each integral entry:
\begin{align}
    \frac{\mathrm{d}}{\mathrm{dt}}\frac{1}{2\pi i} \oint \frac{u_i^{2j+i-2}\exp(2t/(u_i^2-1))}{(u_i-1)^K}\mathrm{du_i} 
    = \frac{1}{2\pi i} \oint \frac{2u_i^{2j+i-2}\exp(2t/(u_i^2-1))}{(u_i-1)^K(u_i^2-1)}\mathrm{du_i} \nonumber \\
    = \frac{2}{2\pi i} \oint \frac{u_i^{2j+i-2}\exp(2t/(u_i^2-1))}{(u_i-1)^{K+1}(u_i+1)}\mathrm{du_i}. \label{deriv} 
\end{align}

By the identity \eqref{detid}, the first derivative of the determinant is given by the sum over determinants where only one column has been differentiated. To account for all these terms, we consider two cases. First, consider the determinant whose first column ($j=1$) is differentiated. Equation \eqref{deriv} tells us that the terms along the first differentiated column are of the form: 
\begin{equation}\label{diff1}
        \frac{1}{\pi i} \oint \frac{u_i^{i}\exp(2t/(u_i^2-1))}{(u_i-1)^{K+1}(u_i+1)}\mathrm{du_i}    \quad  1 \leq i \leq K.
\end{equation} 
Note that since we carried out the scaling at \eqref{eq23} the contours of integration now contain $\pm 1$.  Since the only poles are at $u_i = \pm 1$, we can enlarge the contour of integration as much as needed, and note that the integrand in \eqref{diff1} is $\mathcal{O}(u_i^{-2})$. Thus the integral vanishes as we enlarge the contour of integration since the integrand shrinks faster than the length of the contour grows. Indeed, every entry along the differentiated first column vanishes, which gives us a 0 determinant and so this term does not contribute to the total sum of the derivative of $\det(B)$. 

Next we consider the terms that come from the determinants whose $j^{th}$ column has been differentiated, where $j>1$. For each of these terms, we show that the differentiated $j^{th}$ column is a linear combination of the first $j-1$ columns, therefore the matrix is not full rank and its determinant must be 0. Indeed, if we fix $j>1$ and sum the first $j-1$ columns, we get:
\begin{align}
   \frac{1}{2\pi i} \sum_{l=1}^{j-1}  \oint \frac{u_i^{2l+i-2}\exp(2t/(u_i^2-1))}{(u_i-1)^{K}}\mathrm{du_i} \nonumber \\
      =\frac{1}{2\pi i} \oint  u_i^i \sum_{l=0}^{j-2}  \frac{u_i^{2l}\exp(2t/(u_i^2-1))}{(u_i-1)^{K}}\mathrm{du_i}\label{lincom}.
\end{align}

Using the identity $\sum_{n=0}^{m-1} x^n = \frac{(x^{m}-1)}{(x-1)}$, \eqref{lincom} is equal to:
\begin{align*}    
 \frac{1}{2\pi i}\oint  u_i^i  \frac{(u_i^{2(j-1)}-1)\exp(2t/(u_i^2-1))}{(u_i-1)^{K}(u_i^2-1)}\mathrm{du_i} 
\end{align*}
\begin{equation}\label{diff}
    =  \frac{1}{2\pi i}\oint   \frac{(u_i^{2j+i -2} - u_i^i)\exp(2t/(u_i^2-1))}{(u_i-1)^{K+1}(u_i+1)}\mathrm{du_i}.    
\end{equation}
The final step is to take the difference between the $j^{th}$ column and the linear combination of the first $j-1$ columns and show that each entry vanishes. Indeed,  twice equation \eqref{diff} - equation \eqref{deriv} is equal to \eqref{diff1}
and so the corresponding determinant vanishes. 

This completes our argument that $\det(B)$ is independent of $t$, since it vanishes under differentiation with respect to $t$. This allows us to fix $t = 0$ for simplicity, so that each integral entry \eqref{integral} can be evaluated as a residue:
\begin{align}
  \frac{1}{2\pi i} \oint \frac{u_i^{2j+i-2}}{(u_i-1)^K} \mathrm{du_i}= {\rm Res}_{u_i=1} \Big(\frac{u_i^{2j+i-2}}{(u_i-1)^K}\Big)= \frac{\mathrm{d}^{K-1}}{\mathrm{du_i}^{K-1}}\frac{u_i^{2j+i-2}}{(K-1)!}\Bigg|_{u_i=1}  \nonumber\\
    = \frac{(2j+i-2)(2j+i-3) \dots (2j+i-(K+2))}{(K-1)!} \nonumber\\
    = \binom{2j+i-2}{K-1},
\end{align}
with the convention that $\binom{n}{m}=0$ if $m>n$.
\end{proof}
\begin{lemma}\label{lem1}
The following determinant of binomial coefficients can be computed explicitly in terms of K and is given by $$ \det\left[ \binom{2j+i-2}{K-1} \right]_{i,j = 1}^K  = (-2)^{\binom{K}{2}}.$$
\end{lemma}
\begin{proof}
We'll compute the determinant by turning the matrix into a Vandermonde matrix \eqref{vandermatrix} through some linear combinations of rows, which don't alter the determinant, and with some row swaps, which assign the sign of the determinant according to the parity of $K$. We begin by writing out the matrix:
\begin{align}
\begin{bmatrix}
\binom{1}{K-1} & \binom{3}{K-1} & \dots & \binom{2K-1}{K-1} \\
\binom{2}{K-1} & \binom{4}{K-1} & \dots & \binom{2K}{K-1} \\
\binom{3}{K-1} & \binom{5}{K-1} & \dots & \binom{2K+1}{K-1} \\ \dots \\ 
\binom{K}{K-1} & \binom{K+2}{K-1} & \dots & \binom{3K-2}{K-1} \\
\end{bmatrix}.
\end{align}

Using Pascal's recurrence 
\begin{align}\label{pascal}
\binom{n}{r} - \binom{n-1}{r} = \binom{n-1}{r-1},
\end{align}
we subtract the $i-1^{th}$ row from the $i^{th}$ for all $i=K, K-1, \ldots,2$, in that order. This process returns the following matrix:
\begin{align}
\begin{bmatrix}
\binom{1}{K-1} & \binom{3}{K-1} & \dots & \binom{2K-1}{K-1} \\
\binom{1}{K-2} & \binom{3}{K-2} & \dots & \binom{2K-1}{K-2} \\
\binom{2}{K-2} & \binom{4}{K-2} & \dots & \binom{2K}{K-2} \\ \dots \\ 
\binom{K-1}{K-2} & \binom{K+1}{K-2} & \dots & \binom{3K-1}{K-2} \\
\end{bmatrix}.
\end{align}

Next, we repeat this process, however now we fix the first and second rows; in other words, we only apply the process to the rows $i=K, K-1,\ldots, 3$. Reiterating this process, each time fixing one more row, we end up with the following matrix, whose determinant is identical to the one we started with:
\begin{align}
\begin{bmatrix}
\binom{1}{K-1} & \binom{3}{K-1} & \dots & \binom{2K-1}{K-1} \\
\binom{1}{K-2} & \binom{3}{K-2} & \dots & \binom{2K-1}{K-2} \\
\binom{1}{K-3} & \binom{3}{K-3} & \dots & \binom{2K-1}{K-3} \\ \dots \\ 
\binom{1}{0} & \binom{3}{0} & \dots & \binom{2K-1}{0} \\
\end{bmatrix} = \left[\binom{2j-1}{K-i}\right]_{i,j= 1}^K.
\end{align}

To see this matrix as a Vandermonde matrix, we work from the bottom up; first, we note that the $K^{th}$ row has each entry equal to $1$ since $\binom{m}{0} = 1$ for all integer $m$. The $K-1^{th}$ row has entries:
\begin{align}
    \binom{2j-1}{1} = 2j-1
\end{align} 
which are linear in the column index $j$. We can add the $K^{th}$ row to the $K-1^{th}$ row and pull out the factor of $2$, so that the bottom two rows now look like
\begin{align}
   2 \begin{bmatrix}
    1 & 2 & 3 & \dots & K \\
    1 & 1 & 1 & \dots & 1
    \end{bmatrix}.
\end{align}

Similarly, the entries in the $K-2^{th}$ row are quadratic polynomials in $j$, since
\begin{align}
    \binom{2j-1}{2} = \frac{(2j-1)(2j-2)}{2} = (2j-1)(j-1) = 2j^2 - 3j + 1.
\end{align}
If we add 3 times the $K-1^{th}$ row and subtract the $K^{th}$ row from the $K-2^{th}$ row, and then pull out the factor of 2, the bottom 3 rows then look like: 
\begin{align}
   2^2 \begin{bmatrix}
    1 & 4 & 9 & \dots & K^2 \\
    1 & 2 & 3 & \dots & K \\
    1 & 1 & 1 & \dots & 1
    \end{bmatrix}.
\end{align}

Similarly, the $K-3^{th}$ row has cubic polynomial entries
\begin{align}
    \binom{2j-1}{3} = \frac{4}{3}j^3 - 4j^2 + \frac{11}{3}j - 1.
\end{align}
We can apply linear combinations of the lower rows and pull the coefficient of the cubic term $\frac{2^3}{3!}$ out of the matrix, such that all entries in the $K-3^{th}$ row are cubes of the column index $j$. We proceed similarly for each row, pulling out the $\frac{2^m}{m!}$ coefficient from the $K-m^{th}$ row, and end up with the following matrix: 
\begin{align}
    \prod_{m=0}^{K-1} \frac{2^m}{m!} \begin{bmatrix}
    1 & 2^{K-1} & \dots & K^{K-1}\\
    1 & 2^{K-2} & \dots & K^{K-2} \\
    \dots \\
    1 & 2 & \dots & K \\
    1 & 1 & \dots & 1 
    \end{bmatrix}.
\end{align}

Now our matrix looks like a Vandermonde matrix, but the degree of the variables is decreasing along columns instead of increasing. To fix this, we rearrange the rows, for a total of $\binom{K}{2}$ row swaps. Now we have a proper Vandermonde matrix whose variables are given by the column indices $j = 1, 2, 3, \dots, K$. We can now conclude that 
\begin{align}
    \det\left[ \binom{2j+i-2}{K-1} \right]_{i,j = 1}^K  = (-1)^{\frac{K(K-1)}{2}}  \Delta(1, 2, \dots, K) \prod_{m=0}^{K-1} \frac{2^m}{m!} \nonumber \\
    = (-1)^{\frac{K(K-1)}{2}} \frac{2^{\binom{K}{2}}G(K+1)}{G(K+1)} \nonumber \\
    = (-2)^{\binom{K}{2}}.
\end{align}
Where $G$ is the Barnes $G$-function, ie $G(K+1) = 0! 1! 2! \dots (K-1)!$
\end{proof}

An identical method to that above allows us to state a slightly more general result (without proof):
\begin{lemma}\label{genlem}
\begin{align}\label{rhs1}
 \det\left[\binom{2j-m}{n-i}\right]_{i,j= 1}^n=(-2)^{\tfrac{n(n-1)}{2}},
\end{align}
where $m=0,1,2$.  Note that the right hand side of \eqref{rhs1} doesn't depend on the value taken by $m$. 
\end{lemma}
We now return to \eqref{det}, which we rewrite here for convenience:
\begin{eqnarray}\label{detagain}
  && (-1)^K \int_{SO(2N)} \left( \frac{ \Lambda^{'}_X}{\Lambda_X }(e^{-\alpha}) \right)^K e^{-K\alpha} \, \mathrm{dX} =  \left(\frac{-1}{2}\right)^{\binom{K}{2}} \left(\frac{N}{a}\right)^{K} \\
 && \qquad \times e^{-Ka} \frac{\mathrm{d}^K}{\mathrm{dt}^K}\Bigg|_{t=0} e^{-at}  \, \det\left[ \frac{1}{2\pi i} \oint \frac{u_i^{2j+i-2}\exp(au_i+2t/(u_i^2-1))}{(u_i-1)^K} \mathrm{du_i} \right]_{i,j = 1}^K \, \left(1+ \mathcal{O}\left( \tfrac{a}{N} \right)\right). \nonumber
 \end{eqnarray}
However since we have shown in the two lemmas above that the term of leading order in $a$ inside the derivative is independent of $t$, and therefore vanishes on differentiation, we now write the expression explicitly including next-to-leading order terms - that is, with an extra factor of the small parameter $a$.  So, instead of the approximation $\exp(au_i) \sim 1$ in the integrals in \eqref{detagain}, which was used in defining $\det(B)$ in Lemma \ref{indep}, we use the approximation $\exp(au_i) = 1 + au_i+O(a^2)$ and consider an expansion inside the derivative in powers of $a$. 

In the determinant in \eqref{detagain} we only need to expand $\exp(au_i)$ to the term $au_i$ in one of the rows in order to obtain the expansion of the determinant to the term linear in $a$. Including the term $au_i$ in row $i$ increases the power of $u_i$ by 1 in the numerator of the integrand, making it equal to the power of $u_i$ in the row $i+1$; this causes the entire $i^{th}$ row to be a scalar multiple of the $i+1^{th}$ row, which means the determinant vanishes. Therefore, we only need to consider the case where we include the $au_i$ term in the $K^{th}$ row.  Then, \eqref{detagain} becomes:
\begin{eqnarray}
 &&(-1)^K \int_{SO(2N)} \left( \frac{ \Lambda^{'}_X}{\Lambda_X }(e^{-\alpha}) \right)^K e^{-K\alpha} \, \mathrm{dX} = \left(\frac{-1}{2}\right)^{\binom{K}{2}} \left(\frac{N}{a}\right)^{K}e^{-Ka} \frac{\mathrm{d}^K}{\mathrm{dt}^K}\Bigg|_{t=0} e^{-at} \, \det(A)\left(1+O(a)\right)\nonumber \\
 &&\quad =\left(\frac{-1}{2}\right)^{\binom{K}{2}} \left(\frac{N}{a}\right)^{K}e^{-Ka}  \sum_{n=0}^K (-1)^n\binom{K}{n}a^n \frac{\mathrm{d}^{K-n}}{\mathrm{d}t^{K-n}}\det(A)\Bigg|_{t=0} \left(1+O(a)\right) \nonumber\\
&& \quad =\left(\frac{-1}{2}\right)^{\binom{K}{2}} \left(\frac{N}{a}\right)^{K}e^{-Ka}\left( \frac{\mathrm{d}^{K}}{\mathrm{d}t^{K}}\det(A) -Ka\frac{\mathrm{d}^{K-1}}{\mathrm{d}t^{K-1}}\det(A)+O(a^2)\right)\Bigg|_{t=0} \label{derivs},
\end{eqnarray} 
where $A$, a matrix with its own dependence on $a$,  has entries 
\begin{align}\label{matA} 
A_{i,j} =  \left\{
     \begin{array}{ll}
\frac{1}{2\pi i} \oint \frac{u_i^{2j+i-2}\exp(2t /(u_i^2-1))}{(u_i-1)^K} \mathrm{du_i} & 1 \leq i \leq K-1, \, 1 \leq j \leq K \\ \\
 \frac{1}{2\pi i} \oint \frac{u_i^{2j+i-2}\exp(2t /(u_i^2-1))}{(u_i-1)^K} \mathrm{du_i}
  + \frac{a}{2\pi i} \oint \frac{u_i^{2j+i-1}\exp(2t /(u_i^2-1))}{(u_i-1)^K} \mathrm{du_i} & i = K, \, 1 \leq j \leq K
   \end{array}
   \right..
 \end{align}

By the multilinearity of the determinant, we can split the determinant of $A$ into two determinants according to the two summands in the $K^{th}$ row. 
That is, we use property \eqref{detid2} to write $\det(A) = \det(B) + a\det(C)$, where $B$ is defined in Lemma \ref{indep} and 
\begin{align}\label{detm}
C =  \left[
     \begin{array}{ll}
       \frac{1}{2\pi i} \oint \frac{u_i^{2j+i-2}\exp(2t /(u_i^2-1))}{(u_i-1)^K} \mathrm{du_i} &  1 \leq i \leq K-1, \,  1 \leq j \leq K,\\ \\
      \frac{1}{2\pi i} \oint \frac{u_i^{2j+i-1}\exp(2t /(u_i^2-1))}{(u_i-1)^K} \mathrm{du_i}  &   i = K, \,  1 \leq j \leq K\\
     \end{array}
   \right].
\end{align}

Thus we have 
\begin{eqnarray}\label{onecase}
 &&(-1)^K \int_{SO(2N)} \left( \frac{ \Lambda^{'}_X}{\Lambda_X }(e^{-\alpha}) \right)^K e^{-K\alpha} \, \mathrm{dX} = \left(\frac{-1}{2}\right)^{\binom{K}{2}} \left(\frac{N}{a}\right)^{K} \\
 &&\qquad \times e^{-Ka}\left[ \frac{\mathrm{d}^{K}}{\mathrm{d}t^{K}}\Big(\det(B)+a\det(C)\Big) -Ka\frac{\mathrm{d}^{K-1}}{\mathrm{d}t^{K-1}}\Big(\det(B)+a\det(C)\Big)+O(a^2)\right]\Bigg|_{t=0} .\nonumber
\end{eqnarray} 
We know from Lemma \ref{indep} that any derivative of $\det(B)$ will be zero, so if $K>1$, the only term contributing at order $a$ in the square bracket is the $K^{th}$ derivative of $a\det(C)$. For $K=1$ we see an additional term $-Ka\det(B)=-Ka (-2)^{\frac{K(K-1)}{2}}=-a$. This is what makes the $K=1$ case different in Theorem \ref{thm2}.

Remark, the matrix $C$ has matrix degree $D(C) = K^2$, therefore $\det(C)$ is a polynomial in $t$ of degree at most $K$, since it does not survive $K+1$ derivatives by Proposition \ref{maxderivatives}. In the Proposition below, we compute the $K^{th}$ derivative of $\det(C)$.
\begin{prop}\label{degk}
The $K^{th}$ derivative of the determinant $\det(C)$ in \eqref{detm} is 
$$ \frac{\mathrm{d}^{K}}{\mathrm{d}t^{K}}\det(C)\Bigg|_{t=0} =  2^{(K^2-K+2)/2} \frac{(2K-3)!!}{(K-1)!}.$$  
\end{prop}
\begin{proof}
We start by noting that $C$ has the correct structure so that we can immediately apply Lemma \ref{multiplicitylemma}.  

Thus
\begin{equation}
     \frac{\mathrm{d}^{K}}{\mathrm{d}t^{K}}\det(C)\Bigg|_{t=0} = \det M_1,
\end{equation}
where
\begin{align}
    M_1 = 2^K
    \begin{bmatrix}
    0 & 0 & \dots & \binom{2K-5}{K-1} & \binom{2K-3}{K-1} \\ \\
    0 & 0& \dots & \binom{2K-4}{K-1} & \binom{2K-2}{K-1} \\ 
     \vdots && \ddots && \vdots\\ 
    0 & \binom{K-1}{K-1} & \dots & \binom{3K-7}{K-1} & \binom{3K-5}{K-1} \\ \\
    \binom{K-1}{K-1} & \binom{K+1}{K-1} & \dots & \binom{3K-5}{K-1} & \binom{3K-3}{K-1} \\
    \end{bmatrix},
\end{align}
and note that the first $K-1$ rows of the first column are all zero. We expand the determinant along the first column to get: 
\begin{align}
    \det(M_1) = (-1)^{K+1}2^K \det(M_2), 
\end{align} 
where $M_2$ is the corresponding $(K-1) \times (K-1)$ minor:
\begin{align}    
M_2 = \begin{bmatrix}
    0 & 0 & \dots & \binom{2K-5}{K-1} & \binom{2K-3}{K-1} \\ \\
    0 & 0 & \dots & \binom{2K-4}{K-1} & \binom{2K-2}{K-1} \\ 
     \vdots & & \ddots && \vdots\\ 
    0 & \binom{K}{K-1} & \dots & \binom{3K-8}{K-1} & \binom{3K-6}{K-1} \\ \\
    \binom{K-1}{K-1} & \binom{K+1}{K-1} & \dots & \binom{3K-7}{K-1} & \binom{3K-5}{K-1}
    \end{bmatrix}.
\end{align}
By rewriting the $0$ entries as $\binom{2j+i-2}{K-1}, M_2$ is identical to a $(K-1) \times (K-1)$ minor of the matrix in Lemma \ref{lem1}, and we can apply the same process of linear combinations using Pascal recursion to obtain the matrix:
\begin{align}\label{m2}
    \Tilde{M_2} = \begin{bmatrix}
    \binom{1}{K-1} & \binom{3}{K-1} & \dots & \binom{2K-5}{K-1} & \binom{2K-3}{K-1} \\ \\
     \binom{1}{K-2} & \binom{3}{K-2} & \dots & \binom{2K-5}{K-2} & \binom{2K-3}{K-2} \\ 
     \vdots & & \ddots && \vdots\\ 
    \binom{1}{2} & \binom{3}{2} & \dots & \binom{2K-5}{2} & \binom{2K-3}{2} \\ \\
    \binom{1}{1} & \binom{3}{1} & \dots & \binom{2K-5}{1} & \binom{2K-3}{1}
    \end{bmatrix} = \left[ \binom{2j-1}{K-i} \right]_{1 \leq i,j \leq K-1}. 
\end{align}
Unlike the matrix in Lemma \ref{lem1}, $\Tilde{M_2}$ does not have a row of $1$s, so to get this matrix into a Vandermonde form, first we need to divide the $j^{th}$ column by $(2j-1)$, for each $j$. All entries are still integers since the term $(2j-1)$ appears in each non-zero entry of the $j^{th}$ column, and now the final row is a row of $1$s. Then,
\begin{align}
    \det(M_2) = \det(\Tilde{M_2}) = & \prod_{j=1}^{K-1}(2j-1) \det \begin{bmatrix}
    0& 0 & \dots & \frac{1}{K-1}\binom{2K-6}{K-2} & \frac{1}{K-1}\binom{2K-4}{K-2} \\ \\
     0 & 0 & \dots & \frac{1}{K-2}\binom{2K-6}{K-3} & \frac{1}{K-2}\binom{2K-4}{K-3} \\ 
     \vdots & & \ddots && \vdots\\ 
    0 & \frac{1}{2}\binom{2}{1} & \dots & \frac{1}{2}\binom{2K-6}{1} & \frac{1}{2}\binom{2K-4}{1} \\ \\
    1 & 1 & \dots & 1 & 1
    \end{bmatrix} \\
    = & \prod_{j=1}^{K-1}(2j-1) \det \left[ \frac{1}{(K-i)}\binom{2j-2}{K-i-1} \right]_{1 \leq i,j \leq K-1}. \label{vmat}
\end{align}
Thus
\begin{align} \label{detM1}
   \det(M_1)& = (-1)^{K+1} 2^K\frac{(2K-3)!!}{(K-1)!} \det\left[ \binom{2j-2}{K-i-1} \right]_{1 \leq i,j \leq K-1} \nonumber \\
    \det(M_1) &= (-1)^{(K^2-K)/2}2^{(K^2-K+2)/2} \frac{(2K-3)!!}{(K-1)!},
\end{align}
by Lemma \ref{genlem}.  
\end{proof}

Putting everything together, we have, for $K \geq 2$:
  \begin{align}\label{allt}
(-1)^K \int_{SO(2N)} \left( \frac{ \Lambda^{'}_X}{\Lambda_X }(e^{-\alpha}) \right)^K e^{-K\alpha} \, \mathrm{dX} \nonumber \\ =   \left(\frac{-1}{2}\right)^{\binom{K}{2}} \left(\frac{N}{a}\right)^{K}e^{-Ka} \frac{\mathrm{d}^K}{\mathrm{dt}^K}\Bigg|_{t=0} \, \left( \det(B)+a\det(C) \right)\left(1+ \mathcal{O}(a)\right) \nonumber \\
=   \left(\frac{-1}{2}\right)^{\binom{K}{2}} \left(\frac{N}{a}\right)^{K}e^{-Ka} a (-1)^{(K^2-K)/2}2^{(K^2-K+2)/2} \frac{(2K-3)!!}{(K-1)!} \left(1+ \mathcal{O}(a)\right) \nonumber \\
= \frac{2N^K}{a^{K-1}}e^{-Ka}  \frac{(2K-3)!!}{(K-1)!}  \left(1+ \mathcal{O}(a)\right).  
\end{align} 
We note that since $a$ is larger than $\tfrac{a}{N}$, we've replaced the error term accordingly. Recalling that at leading order, both $e^{-K\alpha}$ and $e^{-Ka}$ $\sim 1$ as $N \to \infty$ we arrive at Theorem \ref{thm2} for $K \geq 2$:
\begin{align}
\int_{SO(2N)} \left( \frac{ \Lambda^{'}_X}{\Lambda_X }(e^{-\alpha}) \right)^K  \, \mathrm{dX} \nonumber \\ 
=  (-1)^K \frac{2N^K}{a^{K-1}}  \frac{(2K-3)!!}{(K-1)!}  \left( 1 + \mathcal{O}(a) \right). 
\end{align} In the case of $K=1$, we recall the extra contribution of $-a$ from \eqref{onecase} and so we have
\begin{align}
    - \int_{SO(2N)} \left( \frac{ \Lambda^{'}_X}{\Lambda_X }(e^{-\alpha}) \right)^1  \, \mathrm{dX} 
= \frac{N}{a}( -a + 2a) (1 + \mathcal{O}(a)), \nonumber \\
\int_{SO(2N)} \left( \frac{ \Lambda^{'}_X}{\Lambda_X }(e^{-\alpha}) \right)^1  \, \mathrm{dX} 
= - N (1+ \mathcal{O}(a)).
\end{align}

\section{Exact formula for K = 1 and K = 2}\label{SEF}
In this section, we use the exact formulas for the moments of the logarithmic derivative of characteristic polynomials averaged over $SO(2N)$ from \cite{EF} to verify our computations for the first and second moments. 
\begin{theorem}[Mason and Snaith, \cite{EF}]\label{EF}
Given a finite set $A$ of complex numbers where $\mathfrak{Re}(\alpha) > 0$ for $\alpha \in A$ and $|A| \leq N$, then $J(A) = J^{*}(A)$ where
\begin{align}
    J(A) = \int_{SO(2N)} \prod_{\alpha \in A}(-e^{-\alpha})\frac{\Lambda_X^{'}}{\Lambda_X}(e^{-\alpha}) \mathrm{dX},\\
    J^{*}(A) = \sum_{D \subseteq A}e^{-2N\sum_{\delta \in D}}(-1)^{|D|}\sqrt{\frac{Z(D,D)Z(D^{-},D^{-})Y(D)}{Y(D^{-})Z^{\dagger}(D^{-},D)^2}} \nonumber \\
    \times \sum_{\substack{A \setminus D = W_1 \cup \dots W_R \\ |W_r| \leq 2}} \prod_{r=1}^R H_D(W_r),
\end{align} and the sum over the $W_r$ is a sum over all distinct set partitions of $A \setminus D$. Where
\begin{align}
    H_D(W) = \left\{
     \begin{array}{ll}
       \left( \sum\limits_{\delta \in D} \frac{z^{'}}{z}(\alpha - \delta) - \frac{z^{'}}{z}(\alpha + \delta) \right) - \frac{z^{'}}{z}(2\alpha) &  W = \{\alpha\} \subset A \setminus D \\
       \left( \frac{z^{'}}{z}\right)^{'}(\alpha + \hat{\alpha})  &  W = \{\alpha, \hat{\alpha}\} \subset A \setminus D\\
       1 & W = \emptyset
     \end{array}
   \right.
\end{align}
and 
\begin{align}
    z(x) = \frac{1}{1 - e^{-x}}, \\
    Y(A) = \prod\limits_{\alpha \in A} z(2 \alpha), \\
    Z(A,B) = \prod\limits_{\substack{\alpha \in A \\ \beta \in B}} z(\alpha + \beta),
\end{align} and the $\dagger$ adds the restriction that the factors $z(0)$ are omitted. 
\end{theorem}
\subsection{The first moment} For the case of $K = 1$, we consider the set $|A| = 1, A = \{ \alpha \}$ in Theorem \ref{EF}. Then, 
\begin{align}
    J(\{\alpha\}) = - \int_{SO(2N)} \frac{\Lambda_X^{'}}{\Lambda_X}(e^{-\alpha})  e^{-\alpha} \mathrm{dX} \\
    J^{*}(\{\alpha\}) = -\frac{z^{'}}{z}(2 \alpha) - e^{-2N\alpha}z(2\alpha). 
\end{align} Letting $ \alpha = a/N$ where $a = o(1)$ as $N \to \infty$, we have
\begin{align}\label{EF1}
    J^{*}(\{\alpha\}) = -\frac{z^{'}}{z}(2a/N) - e^{-2a}z(2a/N).
\end{align} Next, we note that
\begin{align}
    z(x) = \frac{1}{x} + \mathcal{O}(1) \quad \text{and} \quad     \frac{z^{'}}{z}(x) = - \frac{1}{x} + \mathcal{O}(1),
\end{align} and so \eqref{EF1} becomes
\begin{align}
    \left( \frac{N}{2a} - e^{-2a}\frac{N}{2a} \right)(1 + \mathcal{O}(a/N)) \nonumber \\
    = \left( \frac{N}{2a} - (1-2a)\frac{N}{2a} \right)\left(1 + \mathcal{O}(a)\right) \nonumber \\
   = N\left( 1 + \mathcal{O}(a)\right),
\end{align} which agrees with Theorem \ref{thm2} for $K = 1$.

\subsection{The second moment}\label{3.2} For the case of $K=2$, we consider the set $A = \{\alpha, \alpha + h\}$, since we will later send $ h \to 0$, corresponding to taking all the $\alpha_j = \alpha$ in the proof of Theorem \ref{thm2}.
\begin{align}
  \lim_{h\rightarrow 0}\left[  -e^{-2N\alpha}z(2\alpha)\left(\frac{z^{'}}{z}(h) - \frac{z^{'}}{z}(2\alpha+h) - \frac{z^{'}}{z}(2\alpha+2h)\right)\right.  \nonumber \\
    -e^{-2N(\alpha+h)}z(2\alpha+2h)\left(\frac{z^{'}}{z}(-h) - \frac{z^{'}}{z}(2\alpha+h) - \frac{z^{'}}{z}(2\alpha)\right) \nonumber \\ \left.+ 
    e^{-2N(2\alpha+h)}\frac{z(2\alpha+h)z(-2\alpha-h)z(2\alpha)z(2\alpha+2h)}{z(-h)z(h)}
    +    \left(\frac{z^{'}}{z}\right)^{'}(2\alpha+h) + \frac{z^{'}}{z}(2\alpha)\frac{z^{'}}{z}(2\alpha+2h)\right]
\end{align}
\begin{align}
  =\lim_{h\rightarrow 0}\left[  -e^{-2N\alpha}z(2\alpha)\left(-\frac{1}{h} +\frac{1}{2} - \frac{z^{'}}{z}(2\alpha) - \frac{z^{'}}{z}(2\alpha)+O(h)\right) \right.\nonumber \\
    -e^{-2N\alpha}\big(1-2Nh +O(h^2)\big)\big(z(2\alpha)+z'(2\alpha)(2h)+O(h^2)\big)\left(\frac{1}{h}+\frac{1}{2} - \frac{z^{'}}{z}(2\alpha) - \frac{z^{'}}{z}(2\alpha)+O(h)\right) \nonumber \\ + \left.
    e^{-4N\alpha}z(-2\alpha)z^3(2\alpha)(h^2)+O(h^3)
    +    \left(\frac{z^{'}}{z}\right)^{'}(2\alpha) + \frac{z^{'}}{z}(2\alpha)\frac{z^{'}}{z}(2\alpha)+O(h)\right]
\end{align}
\begin{align}
    =\lim_{h\rightarrow 0}\left[ -e^{-2N\alpha}z(2\alpha)\left(-\frac{1}{h}+\frac{1}{2} - 2\frac{z^{'}}{z}(2\alpha) +O(h)\right)\right. \nonumber \\ -e^{-2N\alpha} \left(z(2\alpha)+2hz'(2\alpha)-2Nhz(2\alpha)+O(h^2)\right)\left(\frac{1}{h}+\frac{1}{2}-2\frac{z^{'}}{z}(2\alpha) +O(h)\right) \nonumber \\\left. +    \left(\frac{z^{'}}{z}\right)^{'}(2\alpha) + \frac{z^{'}}{z}(2\alpha)\frac{z^{'}}{z}(2\alpha)+O(h)\right]  
\end{align}
\begin{align}
    = \left(\frac{z^{'}}{z}\right)^{'}(2\alpha) + \frac{z^{'}}{z}(2\alpha)\frac{z^{'}}{z}(2\alpha) + e^{-2N\alpha}z(2\alpha)\big( -1 + 4\frac{z^{'}}{z}(2\alpha) + 2N + O(h) \big) -2e^{-2N\alpha}z^{'}(2\alpha) \label{125}
    \end{align}

Now we let $\alpha=a/N$, where $a\rightarrow 0$ as $N\rightarrow \infty$ and \eqref{125} becomes 
    \begin{eqnarray}
 &&   \left(\frac{z^{'}}{z}\right)^{'}(2a/N) + \frac{z^{'}}{z}(2a/N)\frac{z^{'}}{z}(2a/N) \nonumber \\
 &&\qquad\qquad\qquad+ e^{-2a}z(2a/N)\big( -1 + 4\frac{z^{'}}{z}(2a/N) + 2N + O(h) \big) -2e^{-2a}z^{'}(2a/N) \nonumber \\
   && = \frac{N^2}{4a^2} + O(1) + \frac{N^2}{4a^2} + O\left(\frac{N}{a}\right) + e^{-2a}\frac{N}{2a}\left( -1 -\frac{2N}{a} + 2N + O(h) \right) + e^{-2a}\left(\frac{N^2}{2a^2} + O(1)\right) \nonumber\\
 &&   = \frac{N^2}{2a^2} + (1 -2a) \left( -\frac{N}{2a} - \frac{N^2}{a^2} + \frac{N^2}{a} + \frac{N^2}{2a^2} + O(a) \right) \nonumber \\
  &&  = \frac{2N^2}{a}\left(1 + O(a) \right),
\end{eqnarray}
which also agrees with Theorem \ref{thm2} for $K = 2$.

\section{The symplectic case}\label{Symplectic} In this section we highlight the differences between the computations for the orthogonal and the symplectic ensembles. The first 3 moments are treated separately. This is because in the symplectic case, we need to expand up to $a^3$ in the small-$a$ approximation of $\exp(au_i)$ in one of the rows of the matrix that plays the same role as \eqref{det}, in order for the determinant to be dependent on $t$.
We begin with the analogue of Proposition \ref{propCFS} for the symplectic ensemble, which is given by:
\begin{prop}{(Conrey, Forrester, Snaith \cite{CFS}, Proposition 2.2)}\label{CFSS} For $N \geq Q$, and $\mathfrak{Re}(\gamma_q) > 0 \quad \forall q$, 
\begin{align}
 \int_{USp(2N)} \frac{\prod_{k=1}^K \Lambda_X(e^{-\alpha_k})}{\prod_{q=1}^Q \Lambda_X (e^{-\gamma_q})}\mathrm{dX} = \sum_{\varepsilon \in \{-1, 1\}^K}e^{N\sum_{k=1}^K (\varepsilon_k\alpha_k)}\frac{\prod\limits_{1 \leq j \leq k \leq K} z(\varepsilon_j\alpha_j + \varepsilon_k\alpha_k)\prod\limits_{1 \leq q < r \leq Q} z(\gamma_q + \gamma_r)}{\prod_{k=1}^K\prod_{q=1}^Q z(\varepsilon_k\alpha_k+\gamma_q) e^{N\sum_{k=1}^K\alpha_k}}.
\end{align}   
\end{prop} We note two differences between  Proposition \ref{CFSS} and Proposition \ref{propCFS} from the orthogonal case.  Both differences appear in the range of the products in the numerator on the right hand side of the identity. The first product on the right hand side ranges from $1 \leq j \leq k \leq K$ in the symplectic case, as opposed to from $ 1 \leq j < k \leq K$ in the orthogonal case; this is why we use \eqref{12} instead of \eqref{11} in our application of Lemma \ref{CFKRSlem}. The second product ranges from $1 \leq q < r \leq Q$ in the symplectic case, as opposed to from $ 1 \leq q \leq r \leq Q$ in the orthogonal case. This means that when we set $\alpha_j = \gamma_q = \alpha=a/N$ for all $j, q$, in the step equivalent to (\ref{eq23}) from the orthogonal case, the first range increases the power on $\left(\frac{N}{a}\right)$ by $K$ and the second range reduces by $K$ the power of $\left(\frac{N}{2a}\right)$, compared to the orthogonal computation. We also get an extra $\frac{1}{2u_i}$ factor for each $1 \leq i \leq K$ coming from the range of the first product in the symplectic case. Then, the $\pm K$ factors of $2$ and of $\left(\frac{N}{a}\right)$ cancel out, and the only difference is one less power of $u_i$ for each $i$. Carrying out the computations analogously, once we've factorised our computation as we did to arrive at  \eqref{det}, we have:
\begin{eqnarray}\label{detsym}
 &&  (-1)^K \int_{USp(2N)} \left( \frac{ \Lambda^{'}_X}{\Lambda_X}(e^{-\alpha}) \right)^K e^{-K\alpha} \, \mathrm{dX}   = \left(\frac{-1}{2}\right)^{\binom{K}{2}}  \left(\frac{N}{a}\right)^{K}\\
  &&\times e^{-Ka} \frac{\mathrm{d}^K}{\mathrm{dt}^K}\Bigg|_{t=0} e^{-at} \, \det\left[ \frac{1}{2\pi i} \oint \frac{u_i^{2j+i-3}\exp(au_i+2t/(u_i^2-1))}{(u_i-1)^K} \mathrm{du_i} \right]_{i,j = 1}^K \, \left( 1 + \mathcal{O}\left( \tfrac{a}{N} \right) \right). \nonumber
\end{eqnarray}
As in the even orthogonal case at \eqref{derivs}, we want to find the leading order term of \eqref{detsym} for large $N$.  Recalling that $a\rightarrow 0$ as $N\rightarrow \infty$, this amounts to finding the leading order term for small $a$ inside the $t$-derivative. 

Defining
\begin{align}\label{Psi}
    \det(\Psi) =: \det\left[ \frac{1}{2\pi i} \oint \frac{u_i^{2j+i-3}\exp(au_i+2t/(u_i^2-1))}{(u_i-1)^K} \mathrm{du_i} \right]_{i,j = 1}^K  ,
\end{align}
we can write
\begin{eqnarray}\label{a-expand}
 &&(-1)^K \int_{USp(2N)} \left( \frac{ \Lambda^{'}_X}{\Lambda_X}(e^{-\alpha}) \right)^K e^{-K\alpha} \, \mathrm{dX}  \nonumber \\
 &&= \left(\frac{-1}{2}\right)^{\binom{K}{2}}  \left(\frac{N}{a}\right)^{K}e^{-Ka}\left[\sum_{n=0}^K (-1)^n a^n \binom{K}{n}\frac{\mathrm{d}^{K-n}}{\mathrm{d}t^{K-n}}\det \Psi \Bigg|_{t=0}\right]\left( 1 + \mathcal{O}\left( \tfrac{a}{N} \right) \right),
\end{eqnarray}
noting that $\det\Psi$ will also have an expansion for small $a$. 

Keeping just the leading order (the constant term in $a$) in the square brackets, \eqref{a-expand} is equal to
\begin{eqnarray}\label{a-expand0}
&& \left(\frac{-1}{2}\right)^{\binom{K}{2}}  \left(\frac{N}{a}\right)^{K}e^{-Ka}\left[\frac{\mathrm{d}^K}{\mathrm{dt}^K}\det \Psi_0  \Bigg|_{t=0}\right]\left( 1 + \mathcal{O}\left(a \right) \right),
\end{eqnarray}
where
\begin{equation}\label{psi0}
\det \Psi_0 =: \det\left[ \frac{1}{2\pi i} \oint \frac{u_i^{2j+i-3}\exp(2t/(u_i^2-1))}{(u_i-1)^K} \mathrm{du_i} \right]_{i,j = 1}^K,
\end{equation}
as we have made the leading-order approximation $\exp(au_i) \sim 1$.

The degree of the matrix $\Psi_0$ is $K(K-2)$ (because the column degree of the $j^{th}$ column is $2j-3$).  Thus by Proposition \ref{maxderivatives} it is independent of $t$. Setting $t=0$ and using the same method as  Lemma \ref{indep},  the determinant above reduces to the determinant of a matrix of binomial coefficients:
\begin{align}
  \det\Psi_0=  \det\left[ \binom{2j+i-3}{K-1 }\right]_{1 \leq i,j \leq K} 
\end{align} with the convention that $\binom{n}{m}=0$ if $m>n$. By applying Pascal's recurrence on the rows, we can simplify this to the determinant:
\begin{align}\label{detpsi0}
  \det\Psi_0=  \det \left[ \binom{2j-2}{K-i} \right]_{1 \leq i,j \leq K}= (-2)^{\binom{K}{2}},
\end{align} 
where the second equality follows from Lemma \ref{genlem}. We can conclude that the leading order term in $a$ of the determinant in \eqref{detsym} vanishes upon differentiation with respect to $t$, just as in the orthogonal case.

Expanding the square brackets to next-to-leading order in $a$ (down to terms linear in $a$), \eqref{a-expand} becomes
\begin{equation}\label{a-expand1}
 \left(\frac{-1}{2}\right)^{\binom{K}{2}}  \left(\frac{N}{a}\right)^{K}e^{-Ka}\left( \frac{\mathrm{d}^K}{\mathrm{dt}^K}\left(\det \Psi_0+a\det \Psi_1\right) -aK\frac{\mathrm{d}^{K-1}}{\mathrm{d}t^{K-1}} \det \Psi_0  \Bigg|_{t=0}  \right) \left( 1 + \mathcal{O}\left(a \right) \right),
\end{equation}
where
\begin{equation}\label{psi1}
\det \Psi_1=\det \left[
     \begin{array}{ll}
       \frac{1}{2\pi i} \oint \frac{u_i^{2j+i-3}\exp(2t /(u_i^2-1))}{(u_i-1)^K} \mathrm{du_i} &  1 \leq i \leq K-1, \,  1 \leq j \leq K,\\ \\
      \frac{1}{2\pi i} \oint \frac{u_i^{2j+i-2}\exp(2t /(u_i^2-1))}{(u_i-1)^K} \mathrm{du_i}  &   i = K, \,  1 \leq j \leq K\\
     \end{array}
   \right].
\end{equation}
The sum $\det\Psi_0+a\det\Psi_1$ is the expansion of $\det \Psi$ down to the next-to-leading order term for small $a$.  The matrix $a\Psi_1$ comes from expanding the factor $\exp(au_i)$ in $\Psi$ (see \eqref{Psi}) to the constant term in rows 1 to $K-1$ and selecting the term linear in $a$ in the $K^{th}$ row. Note that if we had retained the term linear in $a$ in any row other than the $K^{th}$, that row would be a multiple of the row below and would yield a zero determinant.  

In \eqref{a-expand1}, the degree of the matrix $\Psi_1$ is $K(K-1)$ (as the column degree of the $j^{th}$ column is $2j-2$).  Thus from Proposition \ref{maxderivatives}, any $t$-derivative of $\det\Psi_1$ is zero.  We already know that any $t$-derivative of $\Psi_0$ is zero.   Thus the only way that \eqref{a-expand1} is non-zero at this order in $a$ is the case of the first moment:
\begin{align}
    K=1: \;\; - \int_{USp(2N)} \left( \frac{ \Lambda^{'}_X}{\Lambda_X}(e^{-\alpha}) \right) e^{-\alpha} \, \mathrm{dX} = \left(\frac{N}{a}\right)e^{-a}(-a)\det \Psi_0\left( 1 + \mathcal{O}\left(a \right) \right) \nonumber \\
    \int_{USp(2N)} \left( \frac{ \Lambda^{'}_X}{\Lambda_X}(e^{-\alpha}) \right) \, \mathrm{dX} = N\left( 1 + \mathcal{O}\left(a \right) \right),
\end{align}
since at leading order, both $e^{-\alpha}$ and $e^{-a}$ are approximated by 1, and $\det\Psi_0=1$ for $K=1$.

Now we assume that $K>1$. Expanding the square brackets to next-to-next-to-leading order in $a$ (that is to $a^2$) turns \eqref{a-expand} into
\begin{eqnarray}\label{a-expand2}
 \left(\frac{-1}{2}\right)^{\binom{K}{2}} && \left(\frac{N}{a}\right)^{K}e^{-Ka}\left( \frac{\mathrm{d}^K}{\mathrm{dt}^K}\left(\det \Psi_0+a\det \Psi_1+ a^2\det\Psi_{1,1}+ \frac{a^2}{2} \det\Psi_{0,2}+ \frac{a^2}{2} \det\Psi_{2,0}\right)\right.\nonumber \\ &&\left.-aK\frac{\mathrm{d}^{K-1}}{\mathrm{d}t^{K-1}}\left( \det \Psi_0 +a\det\Psi_1\right)+a^2\binom{K}{2} \frac{\mathrm{d}^{K-2}}{\mathrm{d}t^{K-2}}\det \Psi_0 \Bigg|_{t=0}  \right) \left( 1 + \mathcal{O}\left(a \right) \right),
\end{eqnarray}
where $\Psi_0$ is defined at \eqref{psi0}, $\Psi_1$ is defined at \eqref{psi1}, and in general, the subscript $(s_1, s_2)$ denotes that in the small-$a$ expansion of $\exp(au_i)$, we have selected the term containing $a^{s_1}$ in the $K-1^{th}$ row, and the term containing $a^{s_2}$ in the $K^{th}$ row, such that:
\begin{equation}\label{psi21}
    \det \Psi_{1,1}=\det \left[
     \begin{array}{ll}
       \frac{1}{2\pi i} \oint \frac{u_i^{2j+i-3}\exp(2t /(u_i^2-1))}{(u_i-1)^K} \mathrm{du_i} &  1 \leq i \leq K-2, \,  1 \leq j \leq K,\\ \\
      \frac{1}{2\pi i} \oint \frac{u_i^{2j+i-2}\exp(2t /(u_i^2-1))}{(u_i-1)^K} \mathrm{du_i}  &   i = K-1 \; \text{and} \; K, \,  1 \leq j \leq K\\
     \end{array}
   \right],
\end{equation}
\begin{equation}\label{psi22}
    \det \Psi_{0,2}=\det \left[
     \begin{array}{ll}
       \frac{1}{2\pi i} \oint \frac{u_i^{2j+i-3}\exp(2t /(u_i^2-1))}{(u_i-1)^K} \mathrm{du_i} &  1 \leq i \leq K-1, \,  1 \leq j \leq K,\\ \\
      \frac{1}{2\pi i} \oint \frac{u_i^{2j+i-1}\exp(2t /(u_i^2-1))}{(u_i-1)^K} \mathrm{du_i}  &   i = K, \,  1 \leq j \leq K\\
     \end{array}
   \right],
\end{equation}
and 
\begin{equation}\label{psi23}
    \det \Psi_{2,0}=\det \left[
     \begin{array}{ll}
       \frac{1}{2\pi i} \oint \frac{u_i^{2j+i-3}\exp(2t /(u_i^2-1))}{(u_i-1)^K} \mathrm{du_i} &  i \neq K-1, \,  1 \leq j \leq K,\\ \\
       \frac{1}{2\pi i} \oint \frac{u_i^{2j+i-1}\exp(2t /(u_i^2-1))}{(u_i-1)^K} \mathrm{du_i} &  i=K-1, \,  1 \leq j \leq K,\\ 
     \end{array}
   \right].
\end{equation}
The sum $\det \Psi_0+a\det \Psi_1+ a^2\det\Psi_{1,1}+ \frac{a^2}{2} \det\Psi_{0,2}+ \frac{a^2}{2} \det\Psi_{2,0}$ is the expansion of $\det \Psi$ down to the next-to-next-to-leading order term for small $a$. The matrix $a^2\Psi_{1,1}$ comes from expanding the factor $\exp(au_i)$ in $\Psi$ (see \eqref{Psi}) to the constant term in rows 1 to $K-2$ and selecting the term linear in $a$ in the $K-1^{th}$ and $K^{th}$ row.  Note that if we had retained the term linear in $a$ in any other rows instead, one of those rows would be a multiple of the row below and would yield a zero determinant. The matrix $\frac{a^2}{2}\Psi_{0,2}$ comes from expanding the factor $\exp(au_i)$ in $\Psi$ to the constant term in rows 1 to $K-1$ and selecting the quadratic term in $a$ in the $K^{th}$ row. Similarly, $\frac{a^2}{2}\Psi_{2,0}$ arises from selecting the quadratic term in the $K-1^{th}$ row.  Note that if we had retained the term quadratic in $a$ in any but the $K-1^{th}$ or $K^{th}$ rows, that row would be a multiple of the row two below and would yield a zero determinant. 

The degree of $\Psi_{1,1}$ is $K(K-1)$ (as the $j$th column degree, $D_j(\Psi_{1,1})$, is $2j-2$). The degree of $\Psi_{0,2}$ is $K^2$ (since $D_j(\Psi_{0,2})=2j-1$) and the secondary degree $\Tilde{D}(\Psi_{0,2}) = K(K-5)$.  The degree of $\Psi_{2,0}$ is $K(K-1)$ (since $D_j(\Psi_{2,0})=2j-2$).  Note that in $\Psi_{2,0}$, the element with the highest degree in each column is in the $K-1^{th}$ row rather than the last row, so this matrix does not quite fit the definition of a matrix in $\mathcal{M}$, but up to an exchange of rows it satisfies the definition and so all the results that hold for matrices $M\in \mathcal{M}$ also hold for $\Psi_{2,0}$.

From the matrix degrees we can see by Proposition \ref{maxderivatives} that $\det\Psi_{1,1}$ and $\det\Psi_{2,0}$ cannot survive any $t$-differentiation, and we already know that that is the case for $\det \Psi_0$ and $\det \Psi_1$. Since $D_1(\Psi_{0,2}) = 1$ and $\Tilde{D}(\Psi_{0,2}) = K(K-3)$, by Proposition \ref{mysterylemma}, $\det(\Psi_{0,2})$ vanishes after $K$  derivatives in \eqref{a-expand2}. Thus the only non-zero contribution in \eqref{a-expand2} is the case of the second moment,
\begin{equation}
   K=2: \;\;    \int_{USp(2N)} \left( \frac{ \Lambda^{'}_X}{\Lambda_X}(e^{-\alpha}) \right)^2  \, \mathrm{dX}=  \frac{-1}{2} \left( \frac{N}{a} \right)^2 e^{-2a} \left[   a^2 \det\Psi_0 \right] \left(1 + \mathcal{O}(a) \right)= N^2 \left( 1 + \mathcal{O}(a) \right).
\end{equation}

Now we assume that $K > 2$ and use similar arguments to compute the coefficient of $\frac{N^K}{a^{K-3}}$ in \eqref{a-expand}. The cubic terms in $a$ in the square brackets in \eqref{a-expand} are
\begin{eqnarray}\label{a-expand3}
 &&\left[ \frac{\mathrm{d}^K}{\mathrm{dt}^K}\left(\frac{a^3}{3!} \left(\det\Psi_{3, 0, 0} + \det\Psi_{0,3,0} + \det\Psi_{0,0,3} \right) \right.\right. \nonumber \\
 && + \frac{a^3}{2} \left( \det\Psi_{0,1,2} + \det\Psi_{0,2,1} + \det\Psi_{1, 0, 2} + \det\Psi_{1,2,0}   + \det\Psi_{2, 0, 1} + \det\Psi_{2,1,0}  \right) + a^3\det\Psi_{1,1,1} \Bigg) \nonumber \\
 &&-aK\frac{\mathrm{d}^{K-1}}{\mathrm{d}t^{K-1}}\left( a^2\det\Psi_{1,1}+ \frac{a^2}{2} \det\Psi_{0,2}+ \frac{a^2}{2} \det\Psi_{2,0}\right) \nonumber \\ &&
  \left. + a^2 \binom{K}{2} \frac{\mathrm{d}^{K-2}}{\mathrm{d}t^{K-2}} (a\det \Psi_1) - a^3 \binom{K}{3} \frac{\mathrm{d}^{K-3}}{\mathrm{d}t^{K-3}} (\det \Psi_0)\Bigg|_{t=0}  \right] ,
\end{eqnarray}
where $\Psi_{s_1, s_2, s_3}$ denotes the matrix whose $K-2^{th}$, $K-1^{th}$ and $K^{th}$ row respectively contain the $a^{s_1}$, $a^{s_2}$ and $a^{s_3}$ term from the expansion of $\exp(au_i)$. Table \ref{tab:psi_expansion} outlines how these matrices contribute to the fourth-from-leading order term. 

\begin{table}
\begin{center}
\begin{tabular}{|l|l|l|}
\hline
Matrix & $\frac{\mathrm{d}^K}{\mathrm{dt}^K}\det \Psi_{s_1,s_2,s_3}$ & Comment \\
\hline 
$\Psi_{0,1, 2}$   & not 0  &  $\det(\Psi_{0,1, 2}) = -\det(\Psi_{0,3, 0} )$\\
$\Psi_{0,2, 1}$   & 0 & $K-1^{th}$ row equals $K^{th}$ row \\
$\Psi_{0,3, 0}$   & not 0  & $\det(\Psi_{0,1, 2}) = -\det(\Psi_{0,3, 0} )$\\
$\Psi_{0,0, 3}$   & 0  &  By Proposition \ref{mysterylemma} since $\Tilde{D}(\Psi_{0,0, 3}) = K(K-3)$ and $D_1(\Psi_{0,0, 3}) = 2$\\
$\Psi_{1,0, 2}$   & 0  & $K-2^{th}$ row equals $K-1^{th}$ row  \\
$\Psi_{1,1, 1}$   & 0  & By Proposition \ref{maxderivatives} since $D(\Psi_{1,1, 1}) = K(K-1)$  \\
$\Psi_{1,2, 0}$   & 0  & By Proposition \ref{maxderivatives} since $D(\Psi_{1,2, 0}) = K(K-1)$  \\
$\Psi_{2,0, 1}$   & 0  & By Proposition \ref{maxderivatives} since $D(\Psi_{2,0, 1}) = K(K-1)$  \\
$\Psi_{2,1, 0}$   & 0 & $K-1^{th}$ row equals $K^{th}$ row \\
$\Psi_{3,0, 0}$   & 0  & By Proposition \ref{maxderivatives} since $D(\Psi_{3,0, 0}) = K(K-1)$ \\

\hline
\end{tabular}    
\end{center}
\caption{\label{tab:psi_expansion} List of the possible contributions to the order $a^3$ term in the small $a$ expansion of $\Psi$ (defined in (\ref{Psi})). The second column indicates whether this term will survive in (\ref{a-expand3}) after $K$ differentiations with respect to $t$.}
\end{table}

Now we are reduced to computing $\frac{\mathrm{d}^K}{\mathrm{dt}^K} \det\Psi_{0,1,2}\big|_{t=0}$, where
\begin{equation}\label{Psi030}
  \det \Psi_{0,1,2}=\det \left[
     \begin{array}{ll}
       \frac{1}{2\pi i} \oint \frac{u_i^{2j+i-3}\exp(2t /(u_i^2-1))}{(u_i-1)^K} \mathrm{du_i} &  1 \leq i \leq K-2, \,  1 \leq j \leq K,\\ \\
       \frac{1}{2\pi i} \oint \frac{u_i^{2j+K-3}\exp(2t /(u_i^2-1))}{(u_i-1)^K} \mathrm{du_i} &  i=K-1, \,  1 \leq j \leq K,\\ \\
      \frac{1}{2\pi i} \oint \frac{u_i^{2j+K-1}\exp(2t /(u_i^2-1))}{(u_i-1)^K} \mathrm{du_i}  &   i = K, \,  1 \leq j \leq K\\
     \end{array}
   \right].
\end{equation}
From Lemma \ref{multiplicitylemma}, we only need to evaluate the determinant that is the result of differentiating each column once. If we differentiate each column once and set $t = 0$, we arrive at the matrix

 \begin{align}\frac{\mathrm{d}^K}{\mathrm{dt}^K} \det\Psi_{0,1,2}\big|_{t=0} = 
 2^K  \det \begin{bmatrix}
    0 & 0 & \dots & \binom{2K-6}{K-1} & \binom{2K-4}{K-1} \\ \\
    0 & 0 & \dots & \binom{2K-5}{K-1} & \binom{2K-3}{K-1} \\ 
     \vdots & & \ddots && \vdots\\ 
    0 & 0 & \dots & \binom{3K-9}{K-1} & \binom{3K-7}{K-1} \\ \\
    0 & \binom{K-1}{K-1} & \dots & \binom{3K-7}{K-1} & \binom{3K-5}{K-1}\\ \\
    \binom{K-1}{K-1} & \binom{K+1}{K-1} & \dots & \binom{3K-5}{K-1} & \binom{3K-3}{K-1} 
    \end{bmatrix}=:\det \Theta,
\end{align}
whose determinant can be computed as we've done throughout the paper. First we expand the determinant along the first column. Then
\begin{align}
  \det   \Theta = (-1)^{K+1} 2^K \det \left[
    \begin{array}{lll}
    \binom{2j+i-3}{K-1} && i \neq K-1  \\ \\
    \binom{2j+K-3}{K-1} && i = K-1 
    \end{array} \right]_{1 \leq i,j \leq K-1}.
\end{align}

We can also expand this determinant along the first column, and we have 
\begin{align}
 \det   \Theta = (-1)^{2K+1}2^K \det \left[
    \binom{2j+i-1}{K-1}  \right]_{1 \leq i,j \leq K-2}.
\end{align}

We apply Pascal's recurrence as at \eqref{pascal}, and obtain
\begin{align}
   \det\Theta&= -2^K \det \left[ \binom{2j}{K-i} \right]_{1 \leq i,j \leq K-2} \nonumber \\&=- 2^K \det \left[ \frac{2j(2j-1)}{(K-i)(K-1-i)}\binom{2j-2}{K-2-i} \right]_{1 \leq i,j \leq K-2}\nonumber\\
   &= - 2^K\prod_{j=1}^{K-2}\frac{2j(2j-1)}{(j+1)j}\det \left[ \binom{2j-2}{K-2-i} \right]_{1 \leq i,j \leq K-2}.
\end{align}

This final determinant can be computed by applying Lemma \ref{genlem}, and we can write
\begin{align}
    \det(\Theta) =-2^K \prod_{j=1}^{K-2}\frac{2(2j-1)}{(j+1)} (-2)^{\frac{(K-2)(K-3)}{2}}\nonumber\\
    =(-1)^{\frac{2+(K-2)(K-3)}{2}} 2^K 2^{K-2} 2^{\frac{K^2 - 5K +6}{2}} \frac{(2K-5)!!}{(K-1)!}\nonumber \\
    = (-1)^{\frac{K(K-5)}{2}} \frac{(2K-5)!!}{(K-1)!} 2^{\frac{K^2-K+2}{2}}.
\end{align}  \\
 Now, for $K \geq 4$, \eqref{a-expand3} is reduced to
\begin{align}
\frac{\mathrm{d}^K}{\mathrm{dt}^K}\left(\frac{a^3}{3!} \det\Psi_{0,3,0}  + \frac{a^3}{2} \det\Psi_{0,1,2} \right) \Bigg|_{t=0} \nonumber \\
= \frac{a^3}{3} \det\Theta = \frac{a^3}{3} (-1)^{\frac{K(K-5)}{2}} \frac{(2K-5)!!}{(K-1)!} 2^{\frac{K^2-K+2}{2}}.
\end{align}
For $K = 3$, we also sum the non-differentiated term 
\begin{align}
     - a^3 \binom{K}{3} \frac{\mathrm{d}^{K-3}}{\mathrm{d}t^{K-3}} (\det \Psi_0) =  8a^3,
\end{align} so we can write the third moment:
\begin{align}
    - \int_{USp(2N)} \left( \frac{ \Lambda^{'}_X}{\Lambda_X}(e^{-3\alpha}) \right)^3 e^{-3\alpha} \, \mathrm{dX}  =  \left(\frac{-1}{2}\right)^{\binom{3}{2}}  \left(\frac{N}{a}\right)^{3}e^{-3a} \left(\frac{-8a^3}{3}  + 8a^3 \right) \left(1 + \mathcal{O}(a) \right), \nonumber \\
K=3:\qquad \qquad\qquad  \qquad\qquad \int_{USp(2N)} \left( \frac{ \Lambda^{'}_X}{\Lambda_X}(e^{-3\alpha}) \right)^3 \mathrm{dX}  =  \frac{2 N^3}{3}  \left(1 + \mathcal{O}(a) \right).
\end{align}
For $K \geq 4$, we've arrived at Theorem \ref{thm3},
\begin{align}
  (-1)^K \int_{USp(2N)} \left( \frac{ \Lambda^{'}_X}{\Lambda_X}(e^{-\alpha}) \right)^K e^{-K\alpha} \mathrm{dX} \nonumber  \\ =   \left(\frac{-1}{2}\right)^{\binom{K}{2}}  \frac{N^K}{a^{K-3}} e^{-Ka} \frac{(-1)^{\frac{K(K-5)}{2}}}{3} \frac{(2K-5)!!}{(K-1)!} 2^{\frac{K^2-K+2}{2}}\left(1 + \mathcal{O}(a) \right), \nonumber \\ 
K\geq 4:\qquad\qquad  \int_{USp(2N)} \left( \frac{ \Lambda^{'}_X}{\Lambda_X}(e^{-\alpha}) \right)^K   \mathrm{dX} 
   =  (-1)^K \frac{2}{3}\frac{N^K}{a^{K-3}}\frac{(2K-5)!!}{(K-1)!}  \left(1 + \mathcal{O}(a) \right).
 \end{align}

For the first two moments we also investigate the next-to-leading order term.  Returning to (\ref{a-expand}) in the case $K=1$,
 \begin{eqnarray}
 -\int_{USp(2N)}\frac{\Lambda_X'}{\Lambda_X}(e^{-\alpha})e^{-\alpha} dX&=&\frac{N}{2} e^{-a}\left[ -a \det\Psi\Big|_{t=0} +\frac{d}{dt}\det \Psi\Big|_{t=0}\right] \left( 1 + \mathcal{O}\left( \tfrac{a}{N} \right) \right) \\
 &=&\frac{N}{2} e^{-a}\left[ -a( \det\Psi_0+a\det\Psi_1)\Big|_{t=0} \right.\nonumber \\
 &&\left.+\frac{d}{dt}(\det \Psi_0+a\det\Psi_1+\frac{a^2}{2}\det\Psi_2)\Big|_{t=0}+\mathcal{O}(a^3)\right] \left( 1 + \mathcal{O}\left( \tfrac{a}{N} \right) \right).\nonumber
 \end{eqnarray}
 Here all determinants are one-dimensional, and we have expanded $\det\Psi$ for small $a$ by including successive terms in the expansion of $\exp(au_i)$ in (\ref{Psi}). In particular
 \begin{eqnarray}
 \det \Psi_0\Big|_{t=0}&=& \det\left[\frac{1}{2\pi i}\oint \frac{1}{u-1} du\right]=1\nonumber\\
 \det \Psi_1\Big|_{t=0}&=& \det \left[\frac{1}{2\pi i} \oint \frac{u}{u-1}du\right]=1\nonumber \\
 \frac{d}{dt}\frac{a^2}{2} \det \Psi_2 \Big|_{t=0} &=& \frac{d}{dt} \frac{a^2}{2} \det\left[\frac{1}{2\pi i} \oint \frac{u^2 \exp(2t/(u^2-1))}{u-1} du\right]\Big|_{t=0} \nonumber \\
 &=&a^2 \det \left[\frac{1}{2\pi i} \oint \frac{u^2}{(u-1)^2(u+1)} du \right]= a^2.
 \end{eqnarray}
From Proposition \ref{maxderivatives}, $\det\Psi_0$ and $\det\Psi_1$ don't survive differentiation. 

So we have
 \begin{eqnarray}
 -\int_{USp(2N)}\frac{\Lambda_X'}{\Lambda_X}(e^{-\alpha})e^{-\alpha} dX&=&\frac{N}{a}e^{-a}(-a +\mathcal{O}(a^3))\left( 1 + \mathcal{O}\left( \tfrac{a}{N} \right) \right)\nonumber \\
&=&\frac{N}{a}(1-a+\mathcal{O}(a^2))(-a +\mathcal{O}(a^3))+\mathcal{O}(a)\nonumber \\
 \int_{USp(2N)}\frac{\Lambda_X'}{\Lambda_X}(e^{-\alpha}) dX&=&N-aN+\mathcal{O}(a^2N)+\mathcal{O}(a).
  \end{eqnarray}
  Note that without further restrictions on the rate of decay of $a$ for large $N$, we can't say which of the two error terms will dominate. 

Similarly for the second moment
 \begin{eqnarray}
 \int_{USp(2N)}&&\left(\frac{\Lambda_X'}{\Lambda_X}(e^{-\alpha})\right)^2e^{-2\alpha} dX\nonumber \\
 &&=-\frac{1}{2}\left(\frac{N}{a}\right)^2e^{-2a}\left[\frac{d^2}{dt^2} \det \Psi\Big|_{t=0} -2a\frac{d}{dt}\det \Psi\Big|_{t=0} +a^2\det \Psi \Big|_{t=0}\right]\left( 1 + \mathcal{O}\left( \tfrac{a}{N} \right) \right) \nonumber
  \\
  &&=-\frac{1}{2}\left(\frac{N}{a}\right)^2e^{-2a}\left[\frac{d^2}{dt^2}\left( \det \Psi_0+a\det \Psi_1+a^2\det \Psi_{11}+\frac{a^2}{2}\det \Psi_{02}+\frac{a^2}{2}\det \Psi_{20}\right.\right.\nonumber \\
  &&\qquad\qquad\left.+\frac{a^3}{3!}\det\Psi_{30}+\frac{a^3}{2!}\det \Psi_{21}+\frac{a^3}{2!} \det\Psi_{12}+\frac{a^3}{3!}\det \Psi_{03} \right)\nonumber \\
  &&\qquad\qquad-2a\frac{d}{dt} \left( \det \Psi_0+a\det \Psi_1+a^2\det \Psi_{11}  +\frac{a^2}{2}\det \Psi_{02} +\frac{a^2}{2} \det \Psi_{20}\right)\nonumber \\
  &&\qquad\qquad+ a^2(\det \Psi_0+a\det\Psi_1)+\mathcal{O}(a^4)\bigg]\Bigg|_{t=0}\left( 1 + \mathcal{O}\left( \tfrac{a}{N} \right) \right),
 \end{eqnarray}
 where the $2\times 2$ matrices $\Psi_0$ and $\Psi_1$ are defined as above at (\ref{psi0}) and (\ref{psi1}) respectively, and $\Psi_{mn}$ is defined by starting with $\Psi_0$ and increasing the power of $u_1$ in the numerator of the integrand for elements of the first row by $m$ and increasing the power of $u_2$ in the second row by $n$. 
 
 By using the column degree and results from Section \ref{theset}, or simply by inspection of these $2\times2$ matrices, we see that none of these terms survive except for:
 \begin{eqnarray}
 \det \Psi_0 \big|_{t=0}&=& \left|\begin{array}{cc} \frac{1}{2\pi i} \oint \frac{1}{(u_1-1)^2}du_1&\frac{1}{2\pi i} \oint \frac{u_1^2}{(u_1-1)^2}du_1\\ \frac{1}{2\pi i} \oint \frac{u_2}{(u_2-1)^2}du_2& \frac{1}{2\pi i} \oint \frac{u_2^3}{(u_2-1)^2}du_2\end{array}\right|=\left|\begin{array}{cc}0&2\\1&3\end{array}\right|=-2\nonumber \\
 \det \Psi_1 \big|_{t=0}&=& \left|\begin{array}{cc} \frac{1}{2\pi i} \oint \frac{1}{(u_1-1)^2}du_1&\frac{1}{2\pi i} \oint \frac{u_1^2}{(u_1-1)^2}du_1\\ \frac{1}{2\pi i} \oint \frac{u_2^2}{(u_2-1)^2}du_2& \frac{1}{2\pi i} \oint \frac{u_2^4}{(u_2-1)^2}du_2\end{array}\right|=\left|\begin{array}{cc}0&2\\2&4\end{array}\right|=-4\nonumber \\
 \frac{d}{dt}\det \Psi_{02} \big|_{t=0}&=&\frac{d}{dt} \left|\begin{array}{cc} \frac{1}{2\pi i} \oint \frac{\exp(2t/(u_1^2-1))}{(u_1-1)^2}du_1&\frac{1}{2\pi i} \oint \frac{u_1^2\exp(2t/(u_1^2-1))}{(u_1-1)^2}du_1\\ \frac{1}{2\pi i} \oint \frac{u_2^3\exp(2t/(u_2^2-1))}{(u_2-1)^2}du_2& \frac{1}{2\pi i} \oint \frac{u_2^5\exp(2t/(u_2^2-1))}{(u_2-1)^2}du_2\end{array}\right|_{t=0}\nonumber \\
 &=&2\left|\begin{array}{cc} \frac{1}{2\pi i} \oint \frac{1}{(u_1-1)^3(u_1+1)}du_1&\frac{1}{2\pi i} \oint \frac{u_1^2}{(u_1-1)^2}du_1\\ \frac{1}{2\pi i} \oint \frac{u_2^3}{(u_2-1)^3(u_2+1)}du_2& \frac{1}{2\pi i} \oint \frac{u_2^5}{(u_2-1)^2}du_2\end{array}\right|\nonumber \\
 &&\qquad\qquad+2    \left|\begin{array}{cc} \frac{1}{2\pi i} \oint \frac{1}{(u_1-1)^2}du_1&\frac{1}{2\pi i} \oint \frac{u_1^2}{(u_1-1)^3(u_1+1)}du_1\\ \frac{1}{2\pi i} \oint \frac{u_2^3}{(u_2-1)^2}du_2& \frac{1}{2\pi i} \oint \frac{u_2^5}{(u_2-1)^3(u_2+1)}du_2\end{array}\right|\nonumber \\
 &=&2\left|\begin{array}{cc}0&1\\2&5\end{array}\right|+2\left|\begin{array}{cc}0&0\\3&4\end{array}\right|=-4.
 \end{eqnarray}
 
 So we have
 \begin{eqnarray}
  \int_{USp(2N)}&&\left(\frac{\Lambda_X'}{\Lambda_X}(e^{-\alpha})\right)^2e^{-2\alpha} dX\nonumber \\
 &&=-\frac{1}{2}\left(\frac{N}{a}\right)^2e^{-2a}(-2a^2+\mathcal{O}(a^4))\left( 1 + \mathcal{O}\left( \tfrac{a}{N} \right) \right)\nonumber \\
 &&=-\frac{1}{2}\left(\frac{N}{a}\right)^2(1-2a+\mathcal{O}(a^2))(-2a^2+\mathcal{O}(a^4))+\mathcal{O}(Na)
 \end{eqnarray}
 and so 
 \begin{eqnarray}
 \int_{USp(2N)}&&\left(\frac{\Lambda_X'}{\Lambda_X}(e^{-\alpha})\right)^2dX=N^2-2N^2a +\mathcal{O}(N^2a^2)+\mathcal{O}(Na).
 \end{eqnarray}

Thus the variance is small:
\begin{eqnarray}
 \int_{USp(2N)}\left(\frac{1}{N}\left(\frac{\Lambda_X'}{\Lambda_X}(e^{-\alpha})\right)-1\right)^2 dX=\mathcal{O}(a^2)+\mathcal{O}(\tfrac{a}{N})
 \end{eqnarray}
 and by Chebyshev's inequality the scaled logarithmic derivative $\frac{1}{N}\frac{\Lambda_X'}{\Lambda_X}(e^{-\alpha})\sim 1$ for almost all matrices in $USp(2N)$ when $N$ is large, corroborating the intuition in Section \ref{interpretation}.

\section{The odd orthogonal case}\label{odd} 
When we compute the moments of the logarithmic derivative of characteristic polynomials averaged over $SO(2N+1)$, unlike the other ensembles, here we get a non-zero coefficient for the leading order term, and the proof is much simpler. As explained in Section \ref{interpretation}, this is due to the guaranteed eigenvalue at 1 in this ensemble. The analogue of Proposition \ref{propCFS} is:
\begin{prop}{(Conrey, Forrester, Snaith, Proposition 2.4 \cite{CFS})}\label{CFSO} For $N \geq Q$, and $\mathfrak{Re}(\gamma_q) \geq 0 \quad \forall q$, 
\begin{align}
 \int_{SO(2N+1)} \frac{\prod_{k=1}^K \Lambda_X(e^{-\alpha_k})}{\prod_{q=1}^Q \Lambda_X (e^{-\gamma_q})}\mathrm{dX}  \nonumber \\
 = \sum_{\varepsilon \in \{-1, 1\}^K} \left( \prod_{j=1}^K \varepsilon_j \right) e^{(N+1/2)\sum_{k=1}^K (\varepsilon_k\alpha_k)}\frac{\prod\limits_{1 \leq j < k \leq K} z(\varepsilon_j\alpha_j + \varepsilon_k\alpha_k)\prod\limits_{1 \leq q \leq r \leq Q} z(\gamma_q + \gamma_r)}{\prod_{k=1}^K\prod_{q=1}^Q z(\varepsilon_k\alpha_k+\gamma_q) e^{(N+1/2)\sum_{k=1}^K\alpha_k}}.
\end{align}   
\end{prop} The two differences from the even orthogonal case are the product $\left(\prod_{j=1}^K \varepsilon_j\right)$ and an extra factor \newline  $\exp\left(\frac{1}{2}\sum_{k=1}^K\varepsilon_k\alpha_k- \alpha_k \right)$. The former implies that we must use equations \eqref{11} and \eqref{14} from Lemma \ref{CFKRSlem}, and we arrive at
\begin{align}
     \int_{SO(2N+1)} \frac{\prod_{k=1}^K \Lambda_X(e^{-\alpha_k})}{\prod_{q=1}^Q \Lambda_X (e^{-\gamma_q})}\mathrm{dX} = \frac{(-1)^{\binom{K}{2}}2^K}{K!(2\pi i)^K} 
      \oint \dots \oint \frac{e^{-(N+1/2)\sum_{k=1}^K \alpha_k}\prod_{j=1}^K \alpha_k}{\prod_{j=1}^K\prod_{k=1}^K (w_k-\alpha_j)(w_k+\alpha_j)} \nonumber \\
     \times \, \,  e^{(N + 1/2)\sum_{k=1}^K w_k} \, \, \frac{\prod\limits_{1 \leq j < k \leq K}z(w_j+w_k)\prod\limits_{1 \leq q \leq r \leq Q}z(\gamma_q + \gamma_r) \Delta^2(w^2)  }{\prod_{k=1}^K\prod_{q=1}^Q z(w_k + \gamma_q) }\mathrm{d}\mathbf{w} \label{130}.
\end{align}
Following the steps of the even orthogonal case, we set $K = Q$, differentiate with respect to all $\alpha_j$s, then set $\alpha_j = \gamma_q = \alpha$ for all $j,q$ to get
\begin{eqnarray}
  &&  (-1)^K \int_{SO(2N+1)} \left(\frac{ \Lambda_X^{'}}{\Lambda_X }(e^{-\alpha})\right)^K e^{-K\alpha}\mathrm{dX} \\
  &&\qquad= \frac{(-1)^{\binom{K}{2}}2^K}{K!(2\pi i)^K} 
      \oint \dots \oint \frac{e^{-(N+1/2)K \alpha}}{\prod_{j=1}^K (w_j^2-\alpha^2)^K}  
     \, \left[ (1-(N+1/2)\alpha) + \sum_{j=1}^K \frac{2\alpha^2}{(w_j^2 - \alpha^2)} \right]^K\nonumber \\
     &&\qquad \times \, e^{(N + 1/2)\sum_{k=1}^K w_k} \, \, \frac{\prod\limits_{1 \leq j < k \leq K}z(w_j+w_k)z(2\alpha)^{\binom{K+1}{2}} \Delta^2(w^2)  }{\prod_{k=1}^K z(w_k + \alpha)^K }\mathrm{d}\mathbf{w}.
\end{eqnarray}

Next we scale our variables by $N$ to compute the asymptotics, and, just as before, set $\alpha = a/N$, $w_j = au_j/N$ where $a = o(1)$ as $N \to \infty$ and use the approximation $z(x) \sim 1/x$ for small $x$. This yields:
\begin{eqnarray}
    &&    (-1)^K \int_{SO(2N+1)} \left(\frac{ \Lambda_X^{'}}{\Lambda_X }(e^{-\alpha})\right)^K e^{-K\alpha}\mathrm{dX} = \frac{(-1)^{\binom{K}{2}}2^K}{K!(2\pi i)^K 2^{K(K+1)/2}} \left( \frac{N}{a} \right)^K e^{-Ka\left(1+\tfrac{1}{2N}\right)}  \nonumber \\
  &&   \qquad   \times  \oint \dots \oint e^{(N+ 1/2)\frac{a}{N}\sum_{k=1}^K u_k} \prod_{j=1}^K \frac{ (u_j + 1)^K} {(u_j^2 -1)^K}\left[ 1-a - \frac{a}{2N}+\sum_{j=1}^K\frac{2}{u_j^2-1}\right]^K\nonumber \\
        &&\qquad\qquad\times  \prod_{1 \leq j < k \leq K} \frac{ (u_k^2 - u_j^2)^2}{(u_k + u_j)}  \mathrm{d}\mathbf{u} \left( 1 + \mathcal{O}(a/N)\right) 
\end{eqnarray}
\begin{eqnarray} && = \frac{(-1)^{\binom{K}{2}}}{K!(2\pi i)^K 2^{\binom{K}{2}}} \left( \frac{N}{a} \right)^K e^{-Ka\left(1+\tfrac{1}{2N}\right)} \oint \dots \oint e^{(N+ 1/2)\frac{a}{N}\sum_{k=1}^K u_k} \nonumber \\
   &&   \qquad  \times  \prod_{j=1}^K \frac{1}{(u_j -1)^K}\left[ 1-a - \frac{a}{2N}+\sum_{j=1}^K\frac{2}{u_j^2-1}\right]^K \Delta(u^2)\Delta(u)  \mathrm{d}\mathbf{u} \left( 1 + \mathcal{O}(a/N)\right).
\end{eqnarray}
We then factorise the result by introducing a parameter $t$ and differentiating with respect to it, analogously to \eqref{27}, rewrite the product of Vandermonde determinants exactly as we did in \eqref{vandet} and bring the factors into the determinant by column. This yields:
\begin{align}\label{eq:odd}
     (-1)^K \int_{SO(2N+1)} \left(\frac{ \Lambda_X^{'}}{\Lambda_X }(e^{-\alpha})\right)^K e^{-K\alpha}\mathrm{dX} \nonumber \\ 
     = \left(\frac{-1}{2}\right)^{\binom{K}{2}} \left( \frac{N}{a} \right)^K e^{-Ka\left(1+\tfrac{1}{2N}\right)} \frac{\mathrm{d}^K}{\mathrm{dt}^K}\Big|_{t=0}e^{t-ta-\frac{ta}{2N}} \nonumber \\
     \times \det \left[ \frac{1}{2\pi i}\oint 
     \frac{\exp(au_i + au_i/(2N) + 2t/(u_i^2-1))   u_i^{2j+i-3}}{(u_i -1)^K} \mathrm{d u_i} 
     \right]_{1 \leq i,j \leq K} \left( 1 + \mathcal{O}(a/N)\right).
\end{align}

As in the symplectic case, Section \ref{Symplectic}, we set 
\begin{equation}
     \det(\Psi) := \det  \left[ \frac{1}{2\pi i}\oint 
     \frac{\exp(au_i + au_i/(2N) + 2t/(u_i^2-1))   u_i^{2j+i-3}}{(u_i -1)^K} \mathrm{d u_i} 
     \right]_{1 \leq i,j \leq K}.
\end{equation}
In this odd orthogonal case we calculate the leading order and next-to-leading order contribution to the moment since in this case the leading order term results directly from the eigenvalue at 1 of $SO(2N+1)$ matrices and so is more straightforward, and perhaps less interesting, than in the other ensembles.   To this end we make the expansion $\det(\Psi)=\det(\Psi_0)+a\det(\Psi_1)+O(a^2)$, where taking the approximation $\exp(au_i + au_i/(2N)) \sim 1$ in the matrix elements yields
\begin{align}
    \det(\Psi_0) := \det \left[ \frac{1}{2\pi i}\oint \frac{\exp(2t/(u_i^2-1)) u_i^{2j+i-3}}{(u_i -1)^K} \mathrm{d u_i} 
     \right]_{1 \leq i,j \leq K} = (-2)^{\binom{K}{2}};
\end{align}
this is the same $\Psi_0$ as at (\ref{detpsi0}).

The determinant $\Psi_1$ is, as in the symplectic case,
\begin{eqnarray}
     && \det(\Psi_1) := \det  \left[ \begin{array}{ll}
        \frac{1}{2\pi i}\oint 
     \frac{\exp(2t/(u_i^2-1))   u_i^{2j+i-3}}{(u_i -1)^K} \mathrm{d u_i}   &  1 \leq i \leq K-1, 1 \leq j \leq K \\ \\
        \frac{1}{2\pi i}\oint 
     \frac{\exp(2t/(u_i^2-1))   u_i^{2j+i-2}}{(u_i -1)^K} \mathrm{d u_i}  & i = K, 1 \leq j \leq K
      \end{array} 
     \right]_{1 \leq i,j \leq K}\\&&\qquad= K(-2)^{\binom{K}{2}},\nonumber
\end{eqnarray} 
where the calculation of the determinant is carried out as follows.  First, by Proposition \ref{maxderivatives} $\det(\Psi_1)$ is independent of $t$ so we can set $t$ to zero.  Then the first column has all entries equal to 0 except for the entry in the $K^{th}$ row, which is equal to $\binom{K}{K-1} = K$. Expanding along this column we have that
\begin{align}
    \det(\Psi_1) =(-1)^{K-1}K \det\left[\binom{2j+i-1}{K-1}\right]_{i,j = 1, \dots, K-1}
    \nonumber \\
    = (-1)^{K-1}K \det\left[\binom{2j}{K-i}\right]_{i,j = 1, \dots, K-1} \nonumber \\
    = (-1)^{K-1}K \prod_{j=1}^{K-1}2j \prod_{i=1}^{K-1}\frac{1}{K-i} \det\left[ \binom{2j-1}{K-1-i}\right]_{i,j = 1, \dots, K-1} \nonumber \\
    = (-1)^{K-1}K 2^{K-1} (-2)^{\binom{K-1}{2}}=K (-2)^{\binom{K}{2}},
\end{align}
where after the first line we perform manipulations as in the proof of Lemma \ref{lem1}  and the determinant in the next to last line is evaluated using Lemma \ref{genlem}.  The determinant in the first line is the same minor as obtained from expanding $\det(\Psi_0)$ in the same way, so unsurprisingly the values of $\det(\Psi_0)$ and $\det(\Psi_1)$ are closely related. 

The derivative in $t$ in (\ref{eq:odd}) expands as 
\begin{eqnarray}\label{detoo}
   && \frac{\mathrm{d}^K}{\mathrm{dt}^K}\Big|_{t=0}e^{t-ta- \frac{ta}{2N}} \det(\Psi_0+a\Psi_1)=\left(1-a- \frac{a}{2N}\right)^K (\det(\Psi_0) +  a \det(\Psi_1) )  \\  &&+ K \left(1-a- \frac{a}{2N}\right)^{K-1}\frac{\mathrm{d}}{\mathrm{d}t}(\det(\Psi_0) +  a \det(\Psi_1) )  + \dots + \frac{\mathrm{d}^K}{\mathrm{dt}^K}(\det(\Psi_0) +  a \det(\Psi_1) )  \Big|_{t=0}.\nonumber
\end{eqnarray}

 As in the symplectic case, $ \det(\Psi_0) $ and $\det (\Psi_1)$ do not survive differentiation. The first term in \eqref{detoo} contributes $\det(\Psi_0)$ to the term of order $\frac{N^K}{a^K}$ in (\ref{eq:odd}).  There are three terms that will give a contribution of order $\frac{N^K}{a^{K-1}}$ in (\ref{eq:odd}) and they also all come from the first term in \eqref{detoo}. The first contribution is $\binom{K}{K-1}(-a)\det(\Psi_0)  = -aK(-2)^{\binom{K}{2}}$, the second is $a\det(\Psi_1) = aK (-2)^{\binom{K}{2}}$. Both of these contributions are multiplied by the approximation $e^{-Ka\left(1+\tfrac{1}{2N}\right)} \sim 1$. The third contribution comes from taking the second order approximation $e^{-Ka\left(1+\tfrac{1}{2N}\right)} \sim 1 - aK$, which contributes $-aK\det(\Psi_0) = -aK(-2)^{\binom{K}{2}}$. 
 
  Therefore, 
\begin{eqnarray}
    &&(-1)^K \int_{SO(2N+1)} \left(\frac{ \Lambda_X^{'}}{\Lambda_X }(e^{-\alpha})\right)^K e^{-K\alpha}\mathrm{dX}  \\&&\qquad\qquad =\left(\frac{-1}{2}\right)^{\binom{K}{2}}  \left( \frac{N}{a} \right)^K  \left[ (-2)^{\binom{K}{2}} - aK(-2)^{\binom{K}{2}}  + \mathcal{O}(a/N) + \mathcal{O}(a^2) \right],\nonumber 
\end{eqnarray} and so we get our result
\begin{eqnarray}\label{eq:SO(2N)result}
   && \int_{SO(2N+1)} \left(\frac{ \Lambda_X^{'}}{\Lambda_X }(e^{-\alpha})\right)^K \mathrm{dX}  \\
     &&\qquad=(-1)^K \left[ \left( \frac{N}{a} \right)^K - \frac{N^K}{a^{K-1}}K \right] + \mathcal{O}\left(\frac{N^{K-1}}{a^{K-1}}\right)  + \mathcal{O}\left( \frac{N^K}{a^{K-2}}  \right). \nonumber
\end{eqnarray}

As mentioned in the interpretation of our results in Section \ref{interpretation}, the leading order term in the moments of the logarithmic derivative for the odd orthogonal ensemble comes from the contribution of the term which corresponds to the eigenvalue at 1; that is, the term $\frac{-1}{1-s}$ in (\ref{SO2N+1logder}) since this term is a pole in the limit $s \longrightarrow 1$. Therefore, it's also of interest to compute the moments of the logarithmic derivative without the contribution from this pole term over the $SO(2N+1)$ ensemble. Using the computations above, we can investigate these moments, and interestingly it seems that the leading order term after the polar term is subtracted grows at least by a factor of $a^2$ or $a/N$ more slowly than the full moment including the polar term (\ref{eq:SO(2N)result}). Throughout this paper we have set $s=e^{-\alpha}=e^{-a/N}$, so in the large $N$ limit the polar term can be written
\begin{align}
   \frac{-1}{1-e^{-\alpha}}= \frac{-1}{1-1 + a/N + \mathcal{O}(a^2/N^2) } = \frac{N}{a}\frac{-1}{\left(1- \tfrac{a}{2N} + \mathcal{O}(a^2/N^2)\right)} \nonumber \\
    = \frac{-N}{a} - \frac{1}{2} + \mathcal{O}\left(\frac{a}{N}\right).
\end{align}
Then the moments without the pole term are given by
\begin{eqnarray}
    && \int_{SO(2N+1)} \left(\frac{ \Lambda_X^{'}}{\Lambda_X }(e^{-\alpha}) + \frac{1}{1-e^{-\alpha}}\right)^K \mathrm{dX} \nonumber \\
    && = \int_{SO(2N+1)} \sum_{m=0}^K \binom{K}{m} \left( \frac{N}{a} + \tfrac{1}{2}+\mathcal{O}\left(\tfrac{a}{N}\right) \right)^m \left(\frac{ \Lambda_X^{'}}{\Lambda_X }(e^{-\alpha}) \right)^{K-m}
     =  \sum_{m=0}^K \binom{K}{m} \left( \frac{N}{a} + \tfrac{1}{2}+\mathcal{O}\left(\tfrac{a}{N}\right) \right)^m \nonumber \\ 
    &&\qquad \times \left[(-1)^{K-m} \left[ \left( \frac{N}{a} \right)^{K-m} -  \frac{N^{K-m}(K-m)}{a^{K-m-1}}  \right]   + \mathcal{O}\left(\frac{N^{K-m-1}}{a^{K-m-1}}\right)  + \mathcal{O}\left( \frac{N^{K-m}}{a^{K-m-2}} \right)\right] \nonumber \\
     &&= \frac{N^K}{a^K}(1-1)^K - \frac{N^K}{a^{K-1}}\left[ \sum_{m=0}^{K} \binom{K}{m}(-1)^{K-m} (K-m) \right] + \mathcal{O}\left(\tfrac{N^{K-1}}{a^{K-1}} \right)+ \mathcal{O}\left(\tfrac{N^{K}}{a^{K-2}} \right)\nonumber \\
     &&=  - \frac{N^K}{a^{K-1}}\left[ \sum_{m=0}^{K-1} K\binom{K-1}{m}(-1)^{K-m}  \right] + \mathcal{O}\left(\tfrac{N^{K-1}}{a^{K-1}} \right)+ \mathcal{O}\left(\tfrac{N^{K}}{a^{K-2}} \right)\nonumber \\
     &&=  K\frac{N^K}{a^{K-1}}( 1-1)^{K-1} + \mathcal{O}\left(\tfrac{N^{K-1}}{a^{K-1}} \right)+ \mathcal{O}\left(\tfrac{N^{K}}{a^{K-2}} \right)\nonumber \\
     &&=  \mathcal{O}\left(\tfrac{N^{K-1}}{a^{K-1}} \right)+ \mathcal{O}\left(\tfrac{N^{K}}{a^{K-2}} \right).
\end{eqnarray}

To continue to even lower order terms we would need to place further restrictions on the rate of decay of $a$ as $N\rightarrow \infty$ in order to specify whether a term with an extra power of $a$ versus one with 
a $1/N$ factor is the next order contribution.  We also note that terms which are a factor of $a/N$ lower than the leading order have been neglected as far back as \eqref{zapp} when we made the approximation to the $z(x)$ function and so incorporating these adds an extra layer of complication to all the ensembles.

\section*{Acknowledgements}
EA would like to thank Ollie Clark for his thoughtful insight for Lemma \ref{lem1}.  We thank the referee for a very careful and helpful reading of this work. 

\bibliographystyle{plain}
\bibliography{Bib}

\end{document}